\documentclass[11pt]{article}
\usepackage{a4wide}
\usepackage[utf8]{inputenc}

\title{The Asymmetric Travelling Salesman Problem in~Sparse Digraphs\thanks{This research is a part of projects that have received funding from the European Research Council (ERC)
		under the European Union's Horizon 2020 research and innovation programme
		Grant Agreement 677651 (Ł. Kowalik).}}

\usepackage{graphicx}
\graphicspath{ {./images/} }
\usepackage{amsfonts,color}
\usepackage{comment}
\usepackage{thmtools}
\usepackage{thm-restate}
\usepackage{mdframed}
\usepackage[absolute]{textpos}
\usepackage{amsthm}
\usepackage{amssymb}
\usepackage{amsmath}
\usepackage{eucal}
\usepackage{todonotes}
\usepackage{xspace}
\usepackage{array}
\usepackage{caption}
\usepackage{chngpage}
\usepackage{multirow}
\usepackage{booktabs}
\usepackage{pdflscape}
\usepackage{tikz}
\usepackage[noadjust]{cite}

\usepackage[vlined,boxruled]{algorithm2e}

\SetKwHangingKw{MyPrint}{print}

\newtheorem{theorem}{Theorem}[section]
\newtheorem{lemma}[theorem]{Lemma}
\newtheorem{corollary}[theorem]{Corollary}

\newtheorem{claim}{Claim}

\newcommand{\heading}[1]{\medskip\noindent{\bf #1.\ }}%

\newcommand{\ProblemFormat}[1]{{\sc #1}}
\newcommand{\ProblemName}[1]{\ProblemFormat{#1}\xspace}

\newcommand{\probATSP}{\ProblemName{ATSP}}
\newcommand{\probTSP}{\ProblemName{TSP}}
\newcommand{\probATSPFull}{\ProblemName{Asymmetric Travelling Salesman Problem}}
\newcommand{\probDirHam}{\ProblemName{Directed Hamiltonicity}}
\newcommand{\probUndirHam}{\ProblemName{Undirected Hamiltonicity}}

\newcommand{\probForcedTSPFull}{\ProblemName{Forced Travelling Salesman Problem}}
\newcommand{\probForcedTSP}{\ProblemName{Forced TSP}}
\newcommand{\probBipartTSP}{\ProblemName{Bipartite Forced Matching TSP}}
\newcommand{\probBFMTSP}{\ProblemName{BFM-TSP}}

\newcommand{\GraphParam}[1]{\mathtt{#1}}
\newcommand{\pw}{\GraphParam{pw}}
\newcommand{\tw}{\GraphParam{tw}}

\newcommand{\ceil}[1]{\lceil{#1}\rceil}
\newcommand{\floor}[1]{\lfloor{#1}\rfloor}

\newcommand{\eps}{\varepsilon}
\renewcommand{\O}{\mathcal{O}}
\newcommand{\Oh}{\mathcal{O}}
\newcommand{\Ohstar}{\mathcal{O}^*}
\newcommand{\Z}{\mathbb{Z}}
\DeclareMathOperator{\avgdeg}{avgdeg}
\DeclareMathOperator{\outdeg}{outdeg}

\DeclareMathOperator{\per}{per}

\newcommand{\type}{{\sf type}}
\newcommand{\inn}{{\sf in}}
\newcommand{\outt}{{\sf out}}
\newcommand{\midd}{{\sf mid}}

\newcommand{\NN}{\mathbb{N}}
\newcommand{\ZZ}{\mathbb{Z}}
\newcommand{\Tree}{\mathcal{T}}

\newcommand{\ignore}[1]{}

\newcommand{\SmallPicture}[2]{%
  \includegraphics[height=#2]{#1}%
  \xspace
}
\newcommand{\montecarlo}{\SmallPicture{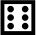}{0.8em}}
\newcommand{\expspace}{\SmallPicture{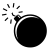}{1em}}

\newcommand{\encircle}[1]{\textcircled{\textsc{#1}}}
\newcommand{\algoref}[1]{Algorithm~\encircle{#1}}
\newcommand{\vect}[1]{\mathbf{#1}}

\begin{document}

\author{Łukasz Kowalik\thanks{Institute of Informatics, University of Warsaw, Poland (\texttt{kowalik@mimuw.edu.pl})}\and Konrad Majewski\thanks{Faculty of Mathematics, Informatics, and Mechanics, University of Warsaw, Poland (\texttt{km371194@students.mimuw.edu.pl})}}

\date{}

\maketitle

\begin{textblock}{20}(0, 13.0)
	\includegraphics[width=40px]{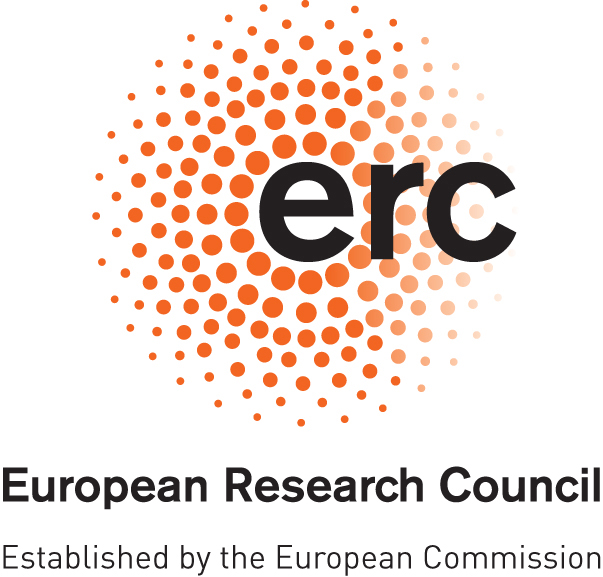}%
\end{textblock}
\begin{textblock}{20}(-0.25, 13.4)
	\includegraphics[width=60px]{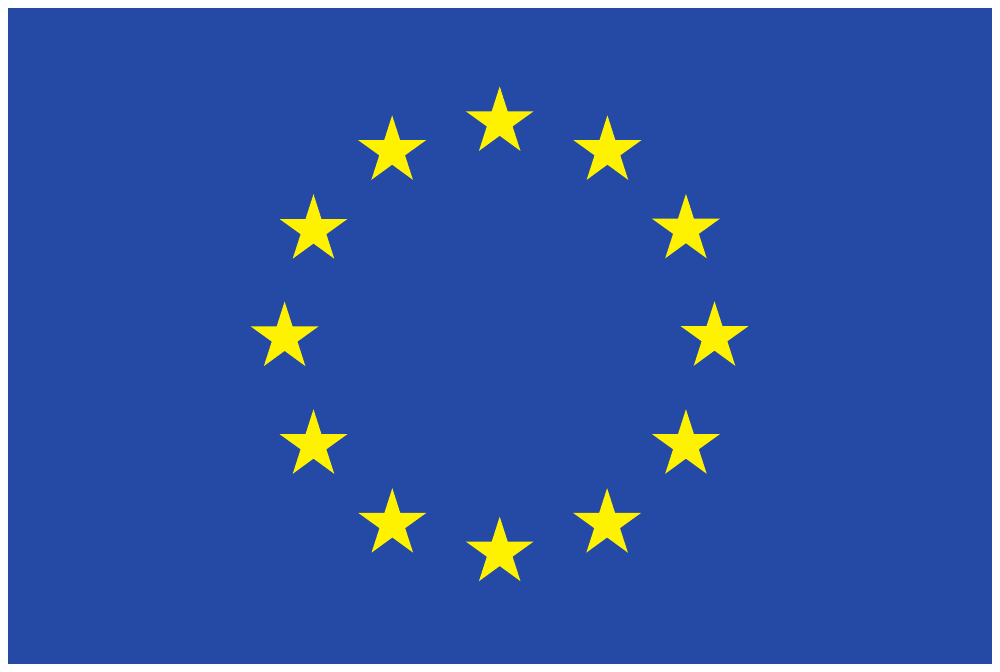}%
\end{textblock}

\begin{abstract}
\probATSPFull (\probATSP) and its special case \probDirHam are among the most fundamental problems in computer science.
	The dynamic programming algorithm running in time $\Ohstar(2^n)$ developed almost 60 years ago by Bellman, Held and Karp, is still the state of the art for both of these problems.

	In this work we focus on {\em sparse} digraphs.

	First, we recall known approaches for \probUndirHam and \probTSP in sparse graphs and we analyse their consequences for \probDirHam and \probATSP in sparse digraphs, either by adapting the algorithm, or by using reductions.
	In this way, we get a number of running time upper bounds for a few classes of sparse digraphs, including $\Ohstar(2^{n/3})$ for digraphs with both out- and indegree bounded by 2, and $\Ohstar(3^{n/2})$ for digraphs with outdegree bounded by 3.

	Our main results are focused on digraphs of bounded {\em average} outdegree $d$.
	The baseline for \probATSP here is a simple enumeration of cycle covers which can be done in time bounded by $\Ohstar(\mu(d)^n)$ for a function $\mu(d)\le(\ceil{d}!)^{1/{\ceil{d}}}$.
	One can also observe that \probDirHam can be solved in randomized time $\Ohstar((2-2^{-d})^n)$ and polynomial space, by adapting a recent result of Bj\"{o}rklund [ISAAC 2018] stated originally for \probUndirHam in sparse bipartite graphs.
	We present two new deterministic algorithms for \probATSP: the first running in time $\Oh(2^{0.441(d-1)n})$ and polynomial space, and the second in exponential space with running time of $\Oh^*(\tau(d)^{n/2})$ for a function $\tau(d)\le d$.
	
\end{abstract}

\newpage

\section{Introduction}
\label{sec:intro}
In the \probDirHam problem, given a  directed graph (digraph) $G$ one has to decide if $G$ has a Hamiltonian cycle, i.e., a simple cycle that visits all vertices.
In its weighted version, called \probATSP, we additionally have integer weights on edges $w:E\to \Z$, and the goal is to find a minimum weight Hamiltonian cycle in $G$.

The \probATSP problem has a dynamic programming algorithm running in time and space $\Ohstar(2^n)$ due to Bellman~\cite{BellmanTSP} and Held and Karp~\cite{HeldKarpTSP}.
Gurevich and Shelah~\cite{GurevichShelahTSP} obtained the best known {\em polynomial space} algorithm, running in time~$\Oh(4^nn^{\log n})$.
It is a major open problem whether there is an algorithm in time $\Ohstar((2-\eps)^n)$ for an $\eps>0$,
even for the unweighted case of \probDirHam.
However, there has been a significant progress in answering this question in variants of \probDirHam.
Namely, Bj\"orklund and Husfeldt~\cite{BH:parity} showed that the parity of the number of Hamiltonian cycles in a digraph can be determined in time $O(1.619^n)$ and Cygan, Kratsch and Nederlof~\cite{Cygan:pathwidth} solved the bipartite case of \probDirHam in time $\Oh(1.888^n)$, which was later improved to $\Ohstar(3^{n/2})=O(1.74^n)$ by Bj\"orklund, Kaski and Koutis~\cite{Bjorklund:directed-bipartite}.

\begin{table}[ht!]
		\begin{center}
		\begin{tabular}{l|>{\raggedleft}p{1.7cm}cl|>{\raggedleft}p{2.5cm}cl}
			\toprule
			Graph class                     & \multicolumn{3}{c|}{\ProblemName{Undirected Hamiltonicity}} & \multicolumn{3}{c}{\ProblemName{Travelling Salesman Problem}} \\
			\midrule
			\multirow{2}{*}{general}        & \multirow{2}{*}{$1.66^n$} & \multirow{2}{*}{\montecarlo} & \multirow{2}{*}{\cite{bjorklund-hamilton}}      & $2^n$            & \expspace & \cite{BellmanTSP, HeldKarpTSP} \\
			&                           &                              &                                                 & $4^n n^{\log n}$ &           & \cite{GurevichShelahTSP} \\ \midrule
			\multirow{2}{*}{bipartite}      & \multirow{2}{*}{$1.42^n$} & \multirow{2}{*}{\montecarlo} & \multirow{2}{*}{\cite{bjorklund-hamilton}}      & $2^n$            & \expspace & \cite{BellmanTSP, HeldKarpTSP} \\
			&                           &                              &                                                 & $4^n$            &           & \cite{OthmanTSP} \\ \midrule
			\multirow{2}{*}{$\Delta = 3$}   & $1.16^n$                  & \expspace\ \montecarlo       & \cite{Cygan:pathwidth}                          & $1.22^n$         & \expspace & \cite{Bodlaender:TSP} \\
			& $1.24^n$                  &                              & \cite{XiaoNagamochi:deg3}                       & $1.24^n$         &           & \cite{XiaoNagamochi:deg3} \\ \midrule
			\multirow{2}{*}{$\Delta = 4$}   & $1.51^n$                  & \expspace\ \montecarlo       & \cite{Cygan:pathwidth}$+$\cite{Fomin:pathwidth} & $1.63^n$         & \expspace & \cite{Bodlaender:TSP}$+$\cite{Fomin:pathwidth} \\
			& $1.59^n$                  & \montecarlo                  & \cite{spotting-trees}                           & $1.70^n$         &           & \cite{XiaoNagamochi:deg4} \\ \midrule
			\multirow{2}{*}{$\Delta = 5$}   & \multirow{2}{*}{$1.63^n$} & \multirow{2}{*}{\montecarlo} & \multirow{2}{*}{\cite{spotting-trees}}          & $1.88^n$         & \expspace & \cite{Bodlaender:TSP}$+$\cite{Fomin:pathwidth} \\
			&                           &                              &                                                 & $2.35^n$         &           & \cite{Yunos:deg5} \\ \midrule
      any $\Delta$                 &  &  &      & $(2-\eps'_\Delta)^n$   & \expspace & \cite{Bjorklund:bounded-degree} \\ \midrule
			\multirow{2}{*}{$\avgdeg \leq d$}&$1.12^{dn}$               & \expspace\ \montecarlo       & \cite{Cygan:pathwidth}$+$\cite{Kneis:pathwidth} & $1.14^{dn}$      & \expspace & \cite{Bodlaender:TSP}$+$\cite{Kneis:pathwidth} \\
      &                           &                              &                                                 & $2^{(1-\eps_d)n}$   & \expspace & \cite{Cygan:average-degree} \\ \midrule
			bipartite & \multirow{2}{*}{$(2-2^{1-d})^{n/2}$} & \multirow{2}{*}{\montecarlo} & \multirow{2}{*}{\cite{Bjorklund:sparse-bipartite}} & & & \\
            $\avgdeg \leq d$ & & & & & & \\ \midrule
			pathwidth                       & $3.42^\pw$                & \expspace\ \montecarlo       & \cite{Cygan:pathwidth}                          & $4.28^\pw$       & \expspace & \cite{Bodlaender:TSP} \\ \midrule
			treewidth                       & $4^\tw$                   & \expspace\ \montecarlo       & \cite{Cygan:connectivity}                       & $9.56^\tw$       & \expspace & \cite{Bodlaender:TSP} \\ \bottomrule

		\end{tabular}
		\caption[Results for \emph{undirected} graphs]
		{\label{table:results-undirected}Running times (with polynomial factors omitted) of
			algorithms for \emph{undirected} graphs. Rows marked with \expspace denote exponential
			space algorithms, rows marked with \montecarlo denote Monte Carlo algorithms.}
		\end{center}
\end{table}

\paragraph{Undirected graphs.}
Even more is known in the undirected setting, where the problems are called \probUndirHam and \probTSP.
Bj\"{o}rklund~\cite{bjorklund-hamilton} shows that \probUndirHam can be solved in time $\Oh(1.66^n)$ in general and $\Ohstar(2^{n/2})=O(1.42^n)$ in the bipartite case.
Very recently, Nederlof~\cite{nederlof:1.9999} showed that the bipartite case of \probTSP admits an algorithm in time $\Oh(1.9999^n)$, assuming that square matrices can be multiplied in time $O(n^{2+o(1)})$.
Finally, there is a number of results for \probUndirHam and \probTSP restricted to graphs that are somewhat sparse.
An early example is an algorithm of Eppstein~\cite{Eppstein:cubic} for \probTSP in graphs of maximum degree 3, running in time $\Ohstar(2^{n/3})=O(1.26^n)$.
This result has been later improved and generalized to larger values of maximum degree, we refer the reader to Table~\ref{table:results-undirected} for details ($\Delta$ denotes the maximum degree).
Perhaps the most general measure of graph sparsity is the average degree~$d$.
Cygan and Pilipczuk~\cite{Cygan:average-degree} showed that whenever $d$ is bounded, the $2^n$ barrier for \probTSP can be broken, although only slightly.
More precisely, they proved the bound $\Ohstar(2^{(1-\eps_d)n})$, where $\eps_d = 1 / (2^{2d+1}\cdot 20d \cdot e^{e^{20d}})$.
We note that although their result was stated for undirected graphs, the same reasoning can be made for digraphs of average {\em total degree} (sum of indegree and outdegree).
For small values of $d$, more significant improvements are possible.
Namely, by combining the algorithms for \probUndirHam and \probTSP parameterized by pathwidth~\cite{Cygan:pathwidth,Bodlaender:TSP} with a bound on pathwidth of sparse graphs~\cite{Kneis:pathwidth} we get the upper bound of $\Oh(1.12^{dn})$ and $\Oh(1.14^{dn})$, respectively. For \probUndirHam, if the input graph is additionally bipartite, Bj\"{o}rklund~\cite{Bjorklund:sparse-bipartite} shows the $\Ohstar((2-2^{1-d})^{n/2})$ upper bound.

\begin{table}[ht!]
		\begin{center}
		\begin{tabular}{l|>{\raggedleft}p{2cm}cl|>{\raggedleft}p{1.9cm}cl}
			\toprule
			\multirow{2}{*}{Graph class}    & \multicolumn{3}{c|}{\multirow{2}{*}{\ProblemName{Directed Hamiltonicity}}} & \multicolumn{3}{c}{\ProblemName{Asymmetric Travelling}} \\
			& \multicolumn{3}{c|}{}                                                      & \multicolumn{3}{c}{\ProblemName{Salesman Problem}} \\
			\midrule
			\multirow{2}{*}{general}        & \multirow{2}{*}{$2^n$}    &                              & \multirow{2}{*}{\cite{Bax:inclusion-exclusion, Karp:inclusion-exclusion, Kohn:inclusion-exclusion}}      & $2^n$            & \expspace & \cite{BellmanTSP, HeldKarpTSP} \\
			&                           &                              &                                                        & $4^n n^{\log n}$ &           & \cite{GurevichShelahTSP} \\ \midrule
			\multirow{2}{*}{bipartite}      & \multirow{2}{*}{$1.74^n$} & \multirow{2}{*}{\montecarlo} & \multirow{2}{*}{\cite{Bjorklund:directed-bipartite}}   & $2^n$            & \expspace & \cite{BellmanTSP, HeldKarpTSP} \\
			&                           &                              &                                                        & $4^n$            &           & \cite{OthmanTSP} \\ \midrule
			$(2,2)$-graphs                & $1.26^n$                  &                              & (Corollary \ref{cor:2-regular})                        & $1.26^n$         &           & (Corollary \ref{cor:2-regular}) \\ \midrule
			$\Delta^+ = 3$    & $1.74^n$             & \expspace                    & (Corollary \ref{cor:d-regular})                        & $1.74^n$         & \expspace & (Corollary \ref{cor:d-regular}) \\ \midrule
			$\Delta = 3$    & $1.13^n$             &                     & (Corollary \ref{cor:totdeg3})                        & $1.13^n$         &  & (Corollary \ref{cor:totdeg3}) \\ \midrule
			any $\Delta$         & $(2-2^{-\Delta/2})^n$ & \montecarlo & (Theorem \ref{thm:hamilton-avgdeg}) & $(2-\eps'_\Delta)^n$   & \expspace & \cite{Bjorklund:bounded-degree} \\ \midrule
      average                         & $\mu(d)^n$                &                              & (Corollary \ref{cor:bregman-avgdeg}) & $\mu(d)^n$            &  & (Corollary \ref{cor:bregman-avgdeg}) \\
			$\outdeg \leq d$                & $2^{0.441(d-1)n}$          &                              & (Theorem \ref{thm:avgdeg})                             & $2^{0.441(d-1)n}$ &           & (Theorem \ref{thm:avgdeg}) \\
			                                & $\sqrt{\tau(d)}^{\ n}$    & \expspace & (Theorem \ref{thm:gebauer_sparse}) & $\sqrt{\tau(d)}^{\ n}$ & \expspace &   (Theorem \ref{thm:gebauer_sparse}) \\
			                                & $(2-2^{-d})^n$            & \montecarlo                  & (Theorem \ref{thm:hamilton-avgdeg})                             & $2^{(1-\eps_{2d})n}$   & \expspace & \cite{Cygan:average-degree} \\
                                      & $2^{(1-\Omega(1/d))n}$    & \montecarlo                    & \cite{bjorklund2020asymptotically}                              &    &  &  \\ \midrule
			treewidth                       & $6^\tw$                   & \expspace\ \montecarlo       & \cite{Cygan:connectivity}                              & &  &  \\ \bottomrule

		\end{tabular}
		\caption[Results for \emph{directed} graphs]
		{\label{table:results-directed}Running times (with polynomial factors omitted) of
			the algorithms for \emph{directed} graphs.
			We preserve the notation from Table \ref{table:results-undirected}.
			By $\Delta^+$ we denote maximum outdegree and $\Delta$ denotes maximum total degree.
			Treewidth refers to the underlying undirected graph.}
		\end{center}
\end{table}

\paragraph{Directed sparse graphs: hidden results}
The goal of this paper is to investigate \probDirHam and \probATSP in {\em sparse directed} graphs.
Quite surprisingly, not much results in this topic are stated explicitly.
In fact, we were able to find just a few references of this kind: Bj\"{o}rklund, Husfeldt, Kaski and Koivisto~\cite{Bjorklund:bounded-degree} describe an algorithm for digraphs with total degree bounded by $D$ that works in time $\Ohstar((2-\eps'_D)^n)$, for $\eps'_D=2-(2^{D+1}-2D-2)^{1/(D+1)}$.
Second, Cygan et al.~\cite{Cygan:connectivity} describe an algorithm for \probDirHam running in time $6^tn^{O(1)}$, where $t$ is the treewidth of the input graph.
Finally, Bj\"{o}rklund and Williams~\cite{Bjorklund:directed-counting} show a~deterministic algorithm which {\em counts} Hamiltonian cycles in directed graphs of average degree $d$ in time $2^{n-\Omega(n/d)}$ and exponential space.
Very recently, Bj\"{o}rklund \cite{bjorklund2020asymptotically}, using a different approach, obtained the same running time for the {\em decision} \probDirHam problem, but lowering the space to polynomial, at the cost of using randomization.
The authors of these two works have not put an effort to optimize the constants hidden in the $\Omega$ notation.
By following the analysis in each of these papers as-is, we get the saving term in the exponent at least $n/(111d)$ (for a faster, randomized algorithm) and $n/(500d)$, respectively.

However, one cannot say that nothing more is known, because many results for undirected graphs imply some running time bounds in the directed setting.
We devote the first part of this work to investigating such implications.
In some cases, the implications are immediate.
For example, Gebauer~\cite{Gebauer:TCS:4-regular,GebauerTSP} shows an algorithm running in time $\Ohstar(3^{n/2})=\Ohstar(1.74^n)$ that solves TSP in graphs of maximum degree 4.
It uses the meet-in-the-middle approach and can be sketched as follows: guess two opposite vertices of the solution cycle, generate a family of paths of length $n/2$ from each of them (of size at most $3^{n/2}$) and store one of the families in a dictionary to enumerate all complementary pairs of paths in time $\Ohstar(3^{n/2})$.
This algorithm, without a change, can be used for ATSP in digraphs of maximum outdegree 3, with the same running time bound (see Theorem~\ref{cor:d-regular}).

The other implications that we found rely on a simple reduction from ATSP to a variant of TSP in {\em bipartite undirected} graphs (see Lemma~\ref{lem:reduce}): replace each vertex $v$ of the input digraph $G$ by two vertices $v^\outt$, $v^\inn$ joined by an edge of weight 0, and for each edge $(u,v)\in E(G)$ create an edge $u^\outt v^\inn$ of the same weight. Then find a lightest Hamiltonian cycle that contains the matching $M=\{v^\outt v^\inn\mid v\in V(G)\}$.
By applying this reduction to a digraph with both outdegrees and indegrees bounded by 2, which we call a $(2,2)$-graph, and using Eppstein's algorithm~\cite{Eppstein:cubic} we get the running time of $\Ohstar(2^{n/3})=\Ohstar(1.26^n)$, see Corollary~\ref{cor:2-regular}.
Another consequence is an algorithm running in time $\Ohstar(2^{n/6})$ for digraphs of maximum total degree $3$, see Corollary~\ref{cor:totdeg3}. These two simple classes of digraphs were studied by Plesn{\'{\i}}k~\cite{plesnik}, who showed that \probDirHam remains NP-complete when restricted to them.

We can also apply the reduction to an arbitrary digraph of average outdegree $d$.
A naive approach would be then to enumerate all perfect matchings in the bipartite graph induced by edges $\{u^\outt v^\inn\mid (u,v)\in E(G)\}$.
Indeed, each such matching corresponds to a cycle cover in the input graph, so we basically enumerate cycle covers and filter-out the disconnected ones.
Thanks to Bregman-Minc inequality~\cite{Bregman:permanent} which bounds the permament in sparse matrices the resulting algorithm has running time $\Ohstar(\mu(d)^n)$, where
  \[
\mu(d) = (\lfloor d \rfloor !)^{\frac{\lfloor d \rfloor + 1 - d}{\lfloor d \rfloor}}
(\lceil d \rceil !)^{\frac{d - \lfloor d \rfloor}{\lceil d \rceil}} \leq (\ceil{d}!)^{1/\ceil{d}}.
\]
See Corollary~\ref{cor:bregman-avgdeg} for details.

Yet another upper bound for digraphs of average outdegree $d$ is obtained by using the reduction described above  and next applying Bj\"{o}rklund's algorithm for sparse bipartite graphs~\cite{Bjorklund:sparse-bipartite} with a slight modification to force the matching $M$ in the Hamiltonian cycle (see Theorem~\ref{thm:hamilton-avgdeg}).
The resulting algorithm has running time $\Ohstar((2-2^{-d})^n)$.

\paragraph{Directed sparse graphs: main results}
The simple consequences that we describe above are complemented by two more technical results.

The first algorithm runs in polynomial space and realizes the following idea.
Assume $d<3$.
Then many of the vertices of the input graph have outdegree at most 2, and we can just branch on vertices of outdegree at least 3, and solve the resulting $(2,2)$-graph using the fast $\Ohstar(2^{n/3})$-time algorithm mentioned before.
This idea can be boosted a bit in the case when the initial branching is too costly, i.e., there are many vertices of high outdegree: then we observe that in such an unbalanced graph one can apply the simple cycle cover enumeration which then runs faster than in graphs of the same density but with balanced outdegrees. After a technical analysis of the running time we get the following theorem.

\begin{theorem}
	\label{thm:avgdeg}
	\probATSP restricted to digraphs of average outdegree at most $d$
	can be solved in time $\Ohstar(2^{\alpha (d - 1) n})$
	and polynomial space, where $\alpha=\tfrac{7}{12}-\tfrac{1}{12(\log_23-1)}< 0.44088$.
\end{theorem}

The second algorithm generalizes Gebauer's meet-in-the-middle approach to digraphs of average outdegree $d$. (We note that it uses exponential space.)

\begin{theorem}
	\label{thm:gebauer_sparse}
	\probATSP restricted to digraphs of average outdegree at~most~$d$
	can be solved in~time~$\Ohstar(\tau(d)^{n/2})$ and the same space,
	where
	\[
	\tau(d) = \lfloor d \rfloor ^ {\lfloor d \rfloor + 1 - d} {(\lfloor d \rfloor + 1) ^ {d - \lfloor d \rfloor}}
	\leq d
	\]
\end{theorem}

\definecolor{blue}{HTML}{1F77B4}
\definecolor{orange}{HTML}{FF7F0E}
\definecolor{green}{HTML}{2CA02C}
\definecolor{red}{HTML}{D62728}

\newcommand{\enumcc}{{\sf\color{blue} enumcc}}
\newcommand{\bjorklund}{{\sf\color{orange} Bj\"orklund}}
\newcommand{\mim}{{\sf\color{red} mim}}
\newcommand{\branch}{{\sf\color{green} branch+}}

\begin{figure}[ht]
	\begin{center}
		\includegraphics[width=12cm]{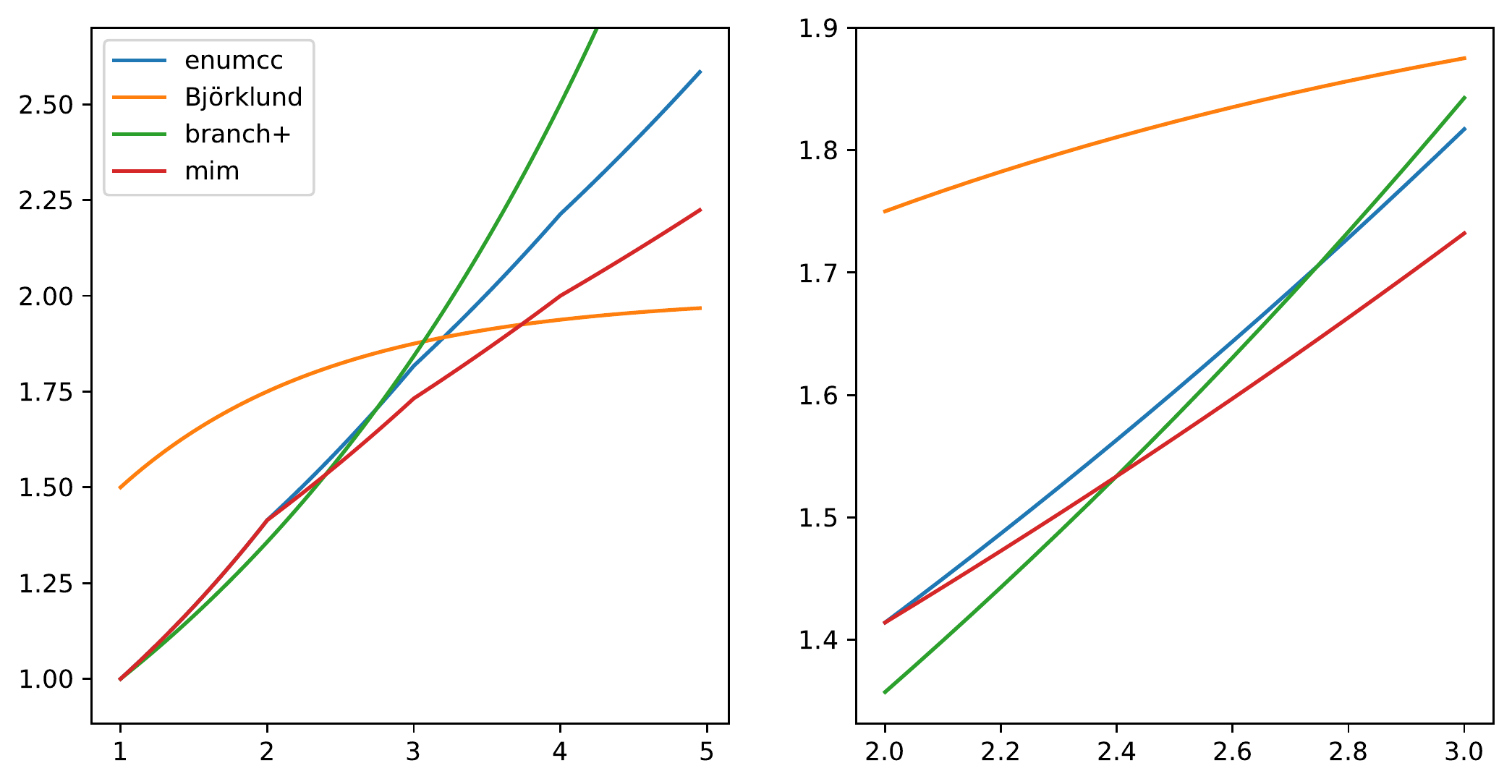}
		\caption{\label{fig:plots}Comparison of the running times of algorithms for solving ATSP (\enumcc, \branch, \mim) and \probDirHam (\bjorklund) in sparse digraphs. Horizontal axis: average degree $d$, vertical axis: base $b$ from the running time bound of the form $\Ohstar(b^n)$.}
	\end{center}
\end{figure}

\paragraph{Which algorithm is the best?}
Figure~\ref{fig:plots} compares four algorithms for solving \probATSP and \probDirHam in digraphs of average outdegree $d$ described above:

\begin{itemize}
	\item \enumcc: enumerating cycle covers (Corollary~\ref{cor:bregman-avgdeg}),
	\item \bjorklund: adaptation of Bj\"orklund's bipartite graphs algorithm  (Theorem~\ref{thm:hamilton-avgdeg}),
	\item \branch: branching boosted by enumerating cycle covers (Theorem~\ref{thm:avgdeg}).
	\item \mim: meet in the middle (Theorem~\ref{thm:gebauer_sparse}),
\end{itemize}

The choice of the best (in terms of the asymptotic worst-case running time) algorithm depends on $d$, on whether we can afford exponential space, and on whether we solve \probATSP or just \probDirHam. We can conclude the following.

\begin{itemize}
	\item {\bf \probATSP in polynomial space:} for $d<2.746$ use \branch, for $d\in[2.746,8.627]$ use \enumcc, and for $d>8.627$ use the general algorithm of Gurevich and Shelah~\cite{GurevichShelahTSP}.
	\item {\bf \probATSP in exponential space:} for $d<2.398$ use \branch, for $d\in[2.398,3.999]$ use \mim, and for $d>3.999$ use the algorithm of Cygan and Pilipczuk \cite{Cygan:average-degree}.
  \item {\bf \probDirHam in polynomial space:} for $d<2.746$ use \branch, for $d\in[2.746,3.203]$ use \enumcc, for $d>3.203$ use \bjorklund, and for sufficiently large~$d$ use the algorithm of Bj\"orklund \cite{bjorklund2020asymptotically}.
	\item {\bf \probDirHam in exponential space:} for $d<2.398$ use \branch, for $d\in[2.398,3.734]$ use \mim, for $d>3.734$ use \bjorklund, and for sufficiently large~$d$ use the algorithm of Bj\"orklund and Williams \cite{Bjorklund:directed-counting}.
\end{itemize}

\section{Reductions from undirected graphs}
\label{sec:reductions}
The objective of this section is to recall two reductions
from the~\probATSP to the~(forced) \probTSP.
Then, we will discuss existing methods of solving \probUndirHam and \probTSP,
and present their implications for corresponding problems
in directed graphs.
The summary of this section is presented in Tables~\ref{table:results-undirected} and~\ref{table:results-directed}.

\subsection{General reductions}

We recall that in the \probForcedTSPFull \cite{Eppstein:cubic, Rubin:forcedTSP, XiaoNagamochi:deg3, XiaoNagamochi:deg4},
we are given an~undirected graph~$G$, a~weight function $w : E(G) \to \Z$,
and a~subset $F \subseteq E(G)$. We say that a~Hamiltonian cycle~$H$ is \emph{admissible},
if~$F \subseteq H$.
The goal is to find an~admissible Hamiltonian cycle of the minimum total weight of the edges
(or, report that there is no such cycle).
Moreover, we define the \probBipartTSP (\probBFMTSP) as a~special case of the \probForcedTSP,
where graph $G$ is bipartite, and the edges of $F$ form a~perfect matching in~$G$.

The following lemma provides the relationship between the \probBFMTSP
and the \probATSP.

\begin{lemma}
\label{lem:reduce}
  For every instance $(G, w)$ of \probATSP,
  where $G$ is a~digraph on $n$~vertices,
  there is an equivalent instance $(\widehat{G}, \widehat{w}, M)$ of \probBFMTSP
  such that $\widehat{G}$~is a~graph on $2n$~vertices.
  
  Moreover, if both outdegrees and indegrees of $G$ are bounded by $D$, then $\widehat{G}$ has maximum degree $D$. Similarly, if $G$ has average outdegree $d$, then $\widehat{G}$ has average degree $d+1$.
\end{lemma}

\begin{proof}
  Let $(G, w)$ be an~instance of~\probATSP. 
  Let $V^\outt=\{v^\outt \mid v\in V(G)\}$ and $V^\inn=\{v^\inn \mid v\in V(G)\}$.
  We define $\widehat{G}$ as a bipartite graph on the vertex set $V(\widehat{G})=V^\outt\cup V^\inn$ with edges $E(\widehat{G})=\{u^\outt v^\inn\mid (u,v)\in E(G)\}\cup M$, where $M$ is the perfect matching $M = \{ v^\inn v^\outt \mid v \in V(G) \}$.
  The edges of~$E(\widehat{G})\setminus M$ inherit the weight from~$G$, i.e. for $(u, v) \in E(G)$ we set
  $\widehat{w}(u^\outt v^\inn) = w(uv)$. Edges of $M$ have weight $0$.
  
  We claim that $(\widehat{G}, \widehat{w}, M)$ is the~desired instance of \probBFMTSP.
  Indeed, $\widehat{G}$~has $2n$~vertices, and given a~Hamiltonian cycle $C := (v_1, \ldots, v_n)$ in~$G$,
  we can construct a~perfect matching~$M' \subseteq E(\widehat{G})$,
  where $M' = \{v_i^\outt v_{i+1}^\inn  \mid i = 1, \ldots, n\}$ (we set ${v_{n+1} := v_1}$). Then, $M \cup M'$ forms
  a~Hamiltonian cycle in $\widehat{G}$ of the same weight as~$C$.
  Conversely, consider a~Hamiltonian cycle~$\widehat{H}$ in~$\widehat{G}$ such that $\widehat{H} = M \cup M'$
  for a~matching~$M'$. Then $M' \subseteq E(\widehat{G}) \setminus M$.
  Hence, after orienting edges of~$M'$ from~$V^\outt$ to~$V^\inn$ and contracting each edge $v^\inn v^\outt \in M$
  to a~single vertex~$v$, we get a~Hamiltonian cycle~$H$ in~$G$ of weight~$\widehat{w}(\widehat{H})$.
\end{proof}

Lemma~\ref{lem:reduce} implies, in particular, that
if there is an~algorithm for \probBFMTSP running in time $\Ohstar(f(n))$,
then there is an~algorithm for \probATSP running in time $\Ohstar(f(2n))$.

When we solve an~\probATSP instance $(G, w)$, in some cases it is easier to work with an~equivalent instance of \probTSP
(without forced edges).

\begin{lemma}
\label{lem:reduce2}
  For every instance $(G, w)$ of \probATSP,
  where $G$ is a~digraph on $n$~vertices,
  there is an~equivalent instance $(\widetilde{G}, \widetilde{w})$ of \probTSP
  such that $\widetilde{G}$~is an~undirected graph on $3n$~vertices.
\end{lemma}

\begin{proof}
This is a classic result.
Given an~instance $(G, w)$ of \probATSP,
we start with constructing an~equivalent instance $(\widehat{G}, \widehat{w}, M)$
of \probBFMTSP by applying Lemma~\ref{lem:reduce}.
Then, we substitute in $\widehat{G}$
every edge $v^\inn v^\outt \in M$ with a simple path of length~$2$: $(v^\inn, v^\midd, v^\outt)$,
where new edges $v^\inn v^\midd$ and~$v^\midd v^\outt$ have weight~$0$.
We see that $\widetilde{G}$~has $3n$~vertices, and every Hamiltonian cycle~$\widetilde{H}$ in~$\widetilde{G}$
corresponds to a~Hamiltonian cycle~$\widehat{H}$ in~$\widehat{G}$ such that
$M \subseteq \widehat{H}$, and $\widehat{w}(\widehat{H}) = \widetilde{w}(\widetilde{H})$.
\end{proof}

\subsection{Enumerating cycle covers}
\label{subsec:reductions-bregman}

Let $(\widehat{G}, \widehat{w}, M)$ be an~instance of \probBFMTSP,
and let $\mathcal{M}$ be a~family of all perfect matchings in $\widehat{G} - M$.
We observe that every cycle cover in $\widehat{G}$ which contains all edges of $M$
is of the form $M \cup M'$, where $M' \in \mathcal{M}$.
Hence, our goal is to find a~matching $M' \in \mathcal{M}$ such that $M \cup M'$ is a~Hamiltonian cycle in~$\widehat{G}$,
and the weight of $M'$ is minimum possible.
One way to do it is to list all the~perfect matchings $M' \in \mathcal{M}$, and choose the best one
among these which form with~$M$ a~Hamiltonian cycle in~$\widehat{G}$.
We will investigate the complexity of such an~approach in sparse graphs.

It is known that all perfect matchings in bipartite graph $\widehat{G}$ can be listed
in time $|\mathcal{M}| n^{\Oh(1)}$ and polynomial space~\cite{Fukuda:matchings}.
Hence, it is enough to provide a~bound on the~size of $\mathcal{M}$ in sparse graphs.
We start with recalling a~classic result of Bregman together with its standard application.

\begin{theorem}[Bregman-Minc inequality~\cite{Bregman:permanent, Schrijver:permanent}]
  Let $A$ be an~$n \times n$ binary matrix, and let $r_i$ denote the number
  of ones in the $i$-th row. Then
  \[
    \per A \leq \prod_{i=1}^n (r_i!)^{1/r_i}.
  \]
\end{theorem}

\begin{corollary}[\cite{ProofsBook}]
\label{cor:bregman-matchings}
  Let $H$ be a~bipartite digraph on $V^\outt \cup V^\inn$, where $|V^\outt| = |V^\inn| = n$,
  and let $d_1, \ldots, d_n$ denote the degrees of vertices of $V^\outt$.
  Then, the number of perfect matchings in $H$ can be bounded by
  \[
    \prod_{i=1}^n (d_i!)^{1/d_i}.
  \]
\end{corollary}

\begin{corollary}
\label{cor:bregman-bounded}
  \probATSP restricted to digraphs of outdegree bounded by $D$
  can be solved in time $(D!)^{n/D} n^{\Oh(1)}$ and polynomial space.
\end{corollary}

\begin{proof}
  Given an~instance $(G, w)$ of \probATSP, we use Lemma~\ref{lem:reduce} to obtain an~equivalent
  instance $(\widehat{G}, \widehat{w}, M)$ of \probBFMTSP. Then, $H := \widehat{G} - M$ is a~bipartite
  graph on $V^\outt \cup V^\inn$, and all vertices of $V^\outt$ in $H$ have degree at most~$D$.
  By Corollary~\ref{cor:bregman-matchings}, there are at most $(D!)^{n/D}$ perfect matchings in~$H$.
  Hence, according to our initial observation, the~instance $(\widehat{G}, \widehat{w}, M)$
  can be solved in time $(D!)^{n/D} n^{\Oh(1)}$.
\end{proof}

To the best of our knowledge, Corollary~\ref{cor:bregman-bounded} provides
the fastest polynomial space algorithm for $D \in \{3, 4, \ldots, 8\}$.
The Bregman-Minc inequality is also useful for digraphs with bounded average
outdegree.
First, we need to quote an analytic lemma.

\begin{lemma}[\cite{AlonSpencer}]
\label{lem:bregman-inequality}
  For a~function $f(d) := (d!)^{1/d}$, and numbers $d_1, d_2 \in \NN$, where $d_1 < d_2$,
  the following inequality holds:
  \[
    f(d_1) f(d_2) \leq f(d_1 + 1) f(d_2 - 1).
  \]
\end{lemma}

\begin{corollary}
\label{cor:bregman-avgdeg}
  \probATSP restricted to digraphs of average outdegree $d$
  can be solved in time $\mu(d)^n n^{\Oh(1)}$ and polynomial space, where
  \[
    \mu(d) = (\lfloor d \rfloor !)^{\frac{\lfloor d \rfloor + 1 - d}{\lfloor d \rfloor}}
    (\lceil d \rceil !)^{\frac{d - \lfloor d \rfloor}{\lceil d \rceil}}
  \]
  In particular, for integral values of $d$, the running time is bounded by $(d!)^{n/d} n^{\Oh(1)}$.
\end{corollary}

\begin{proof}
  As before, we start by constructing an~equivalent instance $(\widehat{G}, \widehat{w}, M)$ of \probBFMTSP.
  Let $d_1, \ldots, d_n$ denote the degrees of vertices of $V^\outt$ in $\widehat{G} - M$.
  Note that their average is equal to $d$.
  By Corollary~\ref{cor:bregman-matchings}, $\widehat{G} - M$ has at most $\prod_{i=1}^n (d_i!)^{1/d_i}$
  perfect matchings.
  Lemma~\ref{lem:bregman-inequality} implies that this value is maximized if for all $i$,
  we have $d_i = \floor{d}$, or~$d_i = \ceil{d}$.
  Then, we claim that there must be $(\floor{d} + 1 - d)n$ vertices of degree~$\floor{d}$
  and $(d - \floor{d})n$ vertices of degree~$\ceil{d}$.
  Indeed, this is true for $d\in\mathbb{N}$, and for $d\not\in\mathbb{N}$, if we denote the number of vertices of degree~$\ceil{d}$ by $\gamma$, then we have $\gamma\ceil{d}+(n-\gamma)\floor{d}=nd$, and hence
  $\gamma=\gamma(\ceil{d}-\floor{d})=n(d-\floor{d})$.
  It follows that there are at most
  $\mu(d)^n$ perfect matchings in $\widehat{G} - M$.
\end{proof}

\subsection{Branching algorithms}
\label{subsec:reductions-branching}

One of the most common techniques which is used for solving NP-hard problems
in sparse graphs is branching (bounded search trees). It is based on optimizing exhaustive search algorithms
by bounding the size of the recursion tree.
In case of \probTSP, the first result of this kind
is due to Eppstein \cite{Eppstein:cubic}.
He showed a~branching algorithm for subcubic graphs running in time $\Ohstar(2^{n/3})$.
Actually, he proved a stronger result in his work.

\begin{theorem}[\cite{Eppstein:cubic}]
\label{thm:eppstein}
  \probForcedTSP restricted to subcubic graphs can be solved
  in time $2^{(n - |F|)/3} n^{\Oh(1)}$ and polynomial space.
\end{theorem}

\begin{corollary}
	\label{cor:2-regular}
	\probATSP restricted to digraphs with all out- and indegrees at most~$2$
	can be solved in time $\Ohstar(2^{n/3})$
	and polynomial space.
\end{corollary}

\begin{proof}
  Let $(G, w)$ be an~instance of \probATSP, where $G$ is a~digraph with all out- and indegrees at most~$2$.
  We apply Lemma~\ref{lem:reduce} to obtain an equivalent instance
  $(\widehat{G}, \widehat{w}, M)$ of \probBFMTSP.
  We know that $\widehat{G}$ has $2n$ vertices, and is subcubic.
  Moreover, $(\widehat{G}, \widehat{w}, M)$ is an instance
  of \probForcedTSP with $|M| = n$ forced edges.
  Hence, we can use Theorem~\ref{thm:eppstein} to solve it
  in time $\Ohstar(2^{(2n - n) / 3}) = \Ohstar(2^{n/3})$.
\end{proof}

We should note here that since the work of Eppstein, faster algorithms
for \probTSP in subcubic graphs were developed \cite{IwamaNakashima:cubic, Liskiewicz:cubic, XiaoNagamochi:deg3}.
However, all of them run still in time $\Ohstar(2^{n/3})$
when we apply them to the $2n$-vertex subcubic graphs resulting from digraphs with all out- and indegrees at most~$2$ (as described in the proof of Corollary~\ref{cor:2-regular}).

Since Lemma~\ref{lem:reduce} allows us
to transfer some of the results for subcubic instances of \probTSP to its version in digraphs  with all out- and indegrees at most~$2$, one may also ask whether there is a relationship between subcubic instances of undirected \probTSP
and instances of \probATSP with maximum total degree at most $3$.
(Recall that total degree of a vertex is the sum of its indegre and outdegree.)
The following lemma (implicit in Plesn{\'{\i}}k~\cite{plesnik}) answers this question indirectly.

\begin{lemma}[\cite{plesnik}]
\label{lem:cubic-dir}
  There is an~algorithm for \probATSP restricted to digraphs of maximum total degree~$3$, 
  and working in time $\O^*(f(n))$
  if and only if there is an~algorithm for \probATSP restricted to digraphs with out- and indegrees at most~$2$,
  and working in time $\O^*(f(2n))$.
\end{lemma}

\begin{proof}
($\implies$) Let $G$ be a~digraph on $n$ vertices with all out- and indegrees at most~$2$.
Let $(\widehat{G},\widehat{w},M)$ be the instance of \probBFMTSP defined in the proof of Lemma~\ref{lem:reduce}.
We construct a weighted digraph $G'$ by orienting the edges of $E(\widehat{G})\setminus M$
from $V^\outt$ to $V^\inn$, and edges of $M$ from $V^\inn$ to $V^\outt$ (the weights stay the same).
We see that $G'$ has $2n$ vertices, has all total degrees at most~3, and Hamiltonian cycles in $G'$
correspond to Hamiltonian cycles in $G$ of the same weight.

($\impliedby$) Let $G_3$ be a digraph on $2n$ vertices with maximum total degree at most~3.
We may assume that all indegrees and outegrees in graph $G_3$ equal to $1$ or $2$, because otherwise $G_3$ has no Hamiltonian cycle.
Since the total degree of each vertex is at most~$3$, each vertex has exactly one incoming edge or exactly one outgoing edge.
We see that every Hamiltonian cycle in $G_3$ must contain all such edges,
hence they can be contracted.
When we contract an edge $(u, v)$ we also remove edges of the form $(u, \underline{\hspace{0.5em}})$ and $(\underline{\hspace{0.5em}}, v)$.
Let us denote the remaining graph by $G'$.
We claim that $G'$ has all out- and indegrees at most~$2$.
Indeed, consider a contraction of an edge $(u,v)$ to a new vertex $x$.
Since we remove all other edges of the form $(u, \underline{\hspace{0.5em}})$ and $(\underline{\hspace{0.5em}}, v)$, there is a one-to-one correspondence between the edges entering (resp. leaving) $x$ and the edges entering $u$ (resp. leaving $v$).
Hence, a single contraction does not increase the maximum out- or indegree, which implies that after all contractions all out- and indegrees are still at most 2. Moreover, every vertex from $G_3$ takes part in at least one edge contraction, and thus
$|V(G')| \leq |V(G_3)|/2 = n$.
\end{proof}

By combining Lemma~\ref{lem:cubic-dir} with Corollary~\ref{cor:2-regular}
we obtain the following.

\begin{corollary}
	\label{cor:totdeg3}
  \probATSP restricted to digraphs of maximum total degree~$3$ can be solved
  in time $\Ohstar(2^{n/6})$ and polynomial space.
\end{corollary}

\subsection{Meet in the middle technique}

Another approach for solving \probTSP in sparse graphs was suggested by Gebauer \cite{GebauerTSP}.
Although it was originally presented for undirected graphs of maximum degree~$4$,
we recall it here for digraphs with outdegrees bounded by~$D$, since the same method can be applied
to them.

\begin{theorem}[\cite{GebauerTSP}]
	\label{cor:d-regular}
	\probATSP restricted to digraphs with outdegrees bounded by~$D$
	can be solved in time $\Ohstar(D^{n/2})$
	and exponential space.
\end{theorem}

The idea of this algorithm can be sketched as follows.
We guess a~pair of vertices $(u_1, u_2)$ which divide a~hypothetical Hamiltonian
cycle into two (almost) equal parts.
Next, we run a~branching procedure to generate all the paths $\mathcal{P}_1$ from~$u_1$ to~$u_2$ of length~$\lfloor n/2 \rfloor$,
and all the paths $\mathcal{P}_2$ from~$u_2$ to~$u_1$ of length~$\lceil n/2 \rceil$.
Finally, we try to combine such paths into one Hamiltonian cycle
by memorizing $\mathcal{P}_1$ in a~dictionary and iterating over paths $P_1 \in \mathcal{P}_1$.

For a~detailed description, we refer to the original work of Gebauer \cite{GebauerTSP}, and to Section~\ref{sec:gebauer},
where we generalize this result to digraphs of bounded \emph{average} outdegree.

\subsection{Algebraic methods}

Bj\"{o}rklund~\cite{Bjorklund:sparse-bipartite} shows the following result.

\begin{theorem}[\cite{Bjorklund:sparse-bipartite}]
	\label{thm:bjorklund-bipartite-sparse}
  There is a Monte Carlo algorithm which solves \probUndirHam
  restricted to bipartite graphs of average degree at most~$d$
  in time $\Ohstar((2 - 2^{1-d})^{n/2})$ and polynomial space.
\end{theorem}

It turns out that the proof of Theorem~\ref{thm:bjorklund-bipartite-sparse} can be modified to get the following Theorem. The idea is to use the reduction of Lemma~\ref{lem:reduce} to get a sparse bipartite graph and modify the construction of Theorem~\ref{thm:bjorklund-bipartite-sparse} so that a relevant forced matching is a part of the resulting Hamiltonian cycle.

\begin{theorem}
  \label{thm:hamilton-avgdeg}
  There is a Monte Carlo algorithm which solves
  \probDirHam restricted to digraphs of average outdegree at most~$d$
  in time $\Ohstar((2-2^{-d})^n)$ and polynomial space.
\end{theorem}

\begin{proof}
  We assume that the reader is familiar with the proof of Theorem~\ref{thm:bjorklund-bipartite-sparse}.
  We apply Lemma~\ref{lem:reduce} and we get a bipartite undirected graph $\widehat{G}=(I\cup J, \widehat{E})$ and a perfect matching $F\subseteq \widehat{E}$. 
  Recall that $\widehat{G}$ has $2n$ vertices and average degree at most $d+1$.
  The goal is to decide whether $\widehat{G}$ has a Hamiltonian cycle $H$ that contains $F$.
  
  Similarly as in~\cite{Bjorklund:sparse-bipartite} we define a polynomial matrix $M$ with rows indexed by the vertices of $I$, and columns indexed by the vertices of $J$, as follows.
  \[M(a,x,z)_{i,j} = \begin{cases}
  \sum_{k\in I\setminus\{i\}}z_{i,j}z_{j,k} (a_{j,k} + x_k) & \text{when $ij\in F$,} \\
  z_{i,j}z_{j,k} (a_{j,k} + x_k) & \text{when $ij\not\in F$, but $jk\in F$.}
  \end{cases}\]
  
  These polynomials have three types of variables: $x_i$ for every $i\in I$, 
  $a_{j,i}$ for every edge $ji\in\widehat{E}$, $j\in J$, $i\in I$.
  The third type of variable is somewhat special.
  Pick a fixed edge $e^*=i^*j^*\in F$.
  For every edge $ij\in\widehat{E}\setminus\{e^*\}$ there is one variable with two names  $z_{i,j}$ and $z_{j,i}$; there are also two different variables $z_{i^*,j^*}$ and $z_{j^*,i^*}$.
  Then we define a polynomial over a large enough field of characteristic two:
  \[\phi=\sum_{x\in\{0,1\}^{n/2}}\det(M(a,x,z))\]
  
  Now we should prove that thanks to cancellation in a field of characteristic two, $\phi=\sum_{H\in\mathcal{H}}\prod_{ij\in H}z_{i,j}$, where $\mathcal{H}$ is the set of all Hamiltonian cycles in $\widehat{G}$ which contain~$F$.
  Bj\"{o}rklund (Lemma 3 in~\cite{Bjorklund:sparse-bipartite}) shows this equality for the original polynomial using three observations: 1) after cancellation, the surviving terms do not contain $a$-variables, 2) each surviving term corresponds to a unique cycle cover in the graph, and 3) terms corresponding to non-Hamiltonian cycle covers pair-up and cancel-out, because if we reverse the lexicographically first cycle that does not contain $e^*$, then we get exactly the same term (and if we reverse a Hamiltonian cycle we get a different term, because of the asymmetry in defining $z$ variables). The arguments used in~\cite{Bjorklund:sparse-bipartite} for proving 1)-3) still hold for the new polynomial, essentially for the same reasons.
  
  The second ingredient of Bj\"{o}rklund's construction is an upper bound on probability that none of the columns of $M(a,x,z)$ is identically zero, where $x\in\{0,1\}^{n/2}$ is a fixed assignment, $z$ is the vector of all $z_{i,j}$ variables, and $a\in\{0,1\}^{n/2}$ is a {\em random} assignment. The calculation relies on the observation that if for a vertex $j\in J$ we have $a_{j,i}+x_i\equiv 0\pmod 2$ for all $ij\in \widehat{E}$, then the column of $j$ is identically zero. Note that this observation still holds for our new design. 
  It follows that the probability bounds derived in~\cite{Bjorklund:sparse-bipartite} apply also in our case. 
  
  The third ingredient is efficient identification of assignments $x\in\{0,1\}^{n/2}$, for which $\det(M(a,x,z))$ is non-zero (for fixed, random, values of $a$). 
  This is done by creating a Boolean variable $w_v$ corresponding to every variable $x_v$ and building a CNF formula such that its satisfying assignments correspond to a superset of all assignments of $x_v$ variables that result in non-zero $\det(M(a,x,z))$.
  Again, the fact that the resulting formula is in CNF follows from the fact that the $j$-th column is non-zero if for some $i\in I$ we have $a_{j,i}+x_i\equiv 1\pmod 2$, which is also true in our design. Finally, Bj\"{o}rklund~\cite{Bjorklund:sparse-bipartite} shows how to enumarate all satisfying assignments of the CNF formula efficiently, what is not altered in any way by our changes in the design of polynomial $\phi$.
\end{proof}

\subsection{Dynamic programming on pathwidth decompositions}

There are many works \cite{Fomin:pathwidth, FominHoie:pathwidth-cubic, Kneis:pathwidth}
which show that the pathwidth of sparse undirected graphs is relatively small,
and which provide a~polynomial time algorithm for computing the~corresponding decomposition.
(For a definition of pathwidth, see \cite{platypus}, section 7.2.)
These results, combined with algorithms working on a~path decomposition of the input graph
\cite{Cygan:pathwidth, Bodlaender:TSP},
often lead to the fastest algorithms for sparse undirected graphs (see Table~\ref{table:results-undirected}).

A natural question that arises here is whether these methods can be transferred
to the corresponding problems in sparse {\em directed} graphs.
There are two natural strategies for that: either use the path decomposition of the underlying undirected graph, or the path decomposition of the graph resulting from the reduction of Lemma~\ref{lem:reduce} or Lemma~\ref{lem:reduce2}.
Although in this way one can get algorithms faster than $\Ohstar(2^n)$ for some classes of sparse digraphs, it does not help to improve any of the bounds in Table~\ref{table:results-directed}, at least by combining currently known results.
For completeness, in the remainder of this section we provide calculations that support this claim. 

Let us try the direct approach first. We can use the following result of Cygan et al.

\begin{theorem}[\cite{Cygan:connectivity}]
	\label{thm:hamilton-treewidth-directed}
	There is a Monte Carlo algorithm which, given a~graph~$G$
	with a tree decomposition of its underlying undirected graph of width $\tw$,
	solves \probDirHam for $G$ in time ${6^{\tw} n^{\Oh(1)}}$
	and exponential space.
\end{theorem}

Consider a $(2,2)$-graph, i.e., a digraph with both out- and indegrees bounded by~$2$.
The undirected graph underlying a $(2,2)$-graph has maximum degree $4$, and hence it has pathwidth at most $n/3+o(n)$, according to Theorem~\ref{thm:pathwidth-degrees} below. 

\begin{theorem}[\cite{Fomin:pathwidth}]
	\label{thm:pathwidth-degrees}
	For every $\eps > 0$, there exists an~integer $N_\eps$ such that for every undirected graph~$G$
	on $n \geq N_\eps$ vertices the inequality
	\[
	\pw(G) \leq \tfrac{1}{6} n_3 + \tfrac{1}{3}n_4 + \tfrac{13}{50} n_5 + n_{\geq 6} + \eps n
	\]
	holds, where $n_k$ is the number of vertices of degree~$k$ in $G$, and $n_{\geq 6}$ is the number
	of vertices of degree at least~$6$.
	Moreover, a~path decomposition which witnesses the above inequality can be computed in polynomial time.
\end{theorem}

This, combined with Theorem~\ref{thm:hamilton-treewidth-directed} gives an algorithm for \probDirHam running in time $\Oh(6^{n/3+o(n)})=\Oh(1.82^n)$, much slower than in Corollary~\ref{cor:2-regular}.

Now consider a digraph of average outdegree degree $d$. 
Then, the underlying undirected graph has average degree $2d$, and we can bound its pathwidth using the following result.

\begin{theorem}[\cite{Kneis:pathwidth}]
	\label{thm:pathwidth-avgdeg}
	Let $G$ be an~$n$-vertex undirected graph of average degree $d$. Then
	\[
	\pw(G) \leq \tfrac{dn}{11.538} + o(n).
	\]
	Moreover, a~path decomposition which witnesses the above inequality can be computed in polynomial time.
\end{theorem}

It follows that the algorithm from Theorem~\ref{thm:hamilton-treewidth-directed} applied on a digraph of average outdegree $d$ has running time of $\Oh(6^{2dn/11.538 + o(n)})=\Oh(1.365^{dn})$, which can be seen to be slower than, say, enumerating cycle covers (Corollary~\ref{cor:bregman-avgdeg}) for all values of $d$.

Now let us focus on the reduction approach. We can use the following two results.

\begin{theorem}[\cite{Cygan:pathwidth}]
\label{thm:hamilton-pathwidth}
  There is a Monte Carlo algorithm which, given a~graph~$G$
  with its path decompositions of width $\pw$,
  solves \probUndirHam for $G$ in time ${(2+\sqrt{2})^{\pw} n^{\Oh(1)}}$
  and exponential space.
  
  Moreover, if $G$ is subcubic, the running
  time is bounded by ${(1+\sqrt{2})^{\pw} n^{\Oh(1)}}$ instead.
\end{theorem}

\begin{theorem}[\cite{Bodlaender:TSP}]
\label{thm:tsp-pathwidth}
  There is an algorithm which, given a~graph~$G$
  with its path decomposition of width $\pw$,
  solves \probTSP for $G$ in time $(2 + 2^{\omega/2})^{\pw} n^{\Oh(1)}$
  and exponential space, where $\omega$ is the matrix multiplication exponent.
  
  Moreover, if $G$ is subcubic, the running
  time is bounded by ${(1 + 2^{\omega/2}) n^{\Oh(1)}}$ instead.
\end{theorem}

Theorems~\ref{thm:hamilton-pathwidth} and~\ref{thm:tsp-pathwidth} combined with Theorem~\ref{thm:pathwidth-degrees}
give, in particular, $\Oh(1.16^n)$ and $\Oh(1.22^n)$ algorithms for \probUndirHam and \probTSP in subcubic graphs, respectively.
For undirected graphs of average degree at most~$d$ we can combine the above theorems with with Theorem~\ref{thm:pathwidth-avgdeg} to obtain algorithms in time $\Oh(1.12^{dn})$ and $\Oh(1.14^{dn})$ for \probUndirHam and \probTSP, respectively.

Now we turn to digraphs again. First consider $(2,2)$-graphs, i.e., digraphs with out- and indegrees bounded by~$2$.
Let $(G, w)$ be an~instance of \probATSP, where $G$ is such a digraph.
We use Lemma~\ref{lem:reduce2} to obtain an~equivalent instance $(\widetilde{G}, \widetilde{w})$
of (undirected) \probTSP. From the construction of $\widetilde{G}$ we see that $\widetilde{G}$
has~$3n$ vertices of which at most~$2n$ have degree~$3$, and the remaining ones have degree~$2$.
Hence, by Theorem~\ref{thm:pathwidth-degrees} we have $\pw(\widetilde{G}) \leq \frac{2n}{6} + o(3n) = \frac{n}{3} + o(n)$.
Therefore, Theorems~\ref{thm:hamilton-pathwidth} and~\ref{thm:tsp-pathwidth} give respectively the algorithms
running in time $(1+\sqrt{2})^{n/3 + o(n)}$ and $(2+2^{\omega/2})^{n/3 + o(n)}$
for \probDirHam and \probATSP in~$G$.
Both results are worse than the running time of the algorithm from Corollary~\ref{cor:2-regular}. 

Again, consider digraphs with bounded average outdegree $d$.
Let $(G, w)$ be an instance of \probATSP, where $G$ is such a~digraph.
We use Lemma~\ref{lem:reduce2} to obtain an~equivalent instance~$(\widetilde{G}, \widetilde{w})$ of \probTSP.
Then, $\widetilde{G}$ has $2n$~vertices of average degree~$d+1$, and $n$~vertices of degree~$2$.
The latter ones can increase the pathwidth only by~$1$ in total, hence
${\pw(\widetilde{G}) \leq \frac{2(d+1)n}{11.538} + o(3n) = \frac{(d+1)n}{5.769}} + o(n)$, and consequently,
\probDirHam and \probTSP in $G$ can be solved in time ${(2+\sqrt{2})^{(d+1)n/5.769 + o(n)}}$
and $(2+2^{\omega/2})^{(d+1)n/5.769 + o(n)}$, respectively.
Both results are worse than the algorithm enumerating cycle covers described in Subsection~\ref{subsec:reductions-bregman}.

\section{Polynomial space algorithm}
\label{sec:avgdeg}
This section is devoted to the proof of Theorem~\ref{thm:avgdeg}.
We begin with introducing some additional notions,
then we provide a~branching algorithm which will be later used
as a~subroutine,
and finally we describe and analyse an~algorithm
for digraphs of average outdegree at most~$d$.

\subsection{Preliminaries}
\heading{Interfaces and switching walks}
Let $G$ be a~directed graph (digraph).
For a vertex $v$ a set $I^\inn_v$ of all incoming edges to $v$
or a set $I^\outt_v$ of all outgoing edges from $v$ will be called
an \emph{interface} of $v$.
We define the type of an interface of $v$ so that $\type(I^\inn_v)=\inn$ and $\type(I^\outt_v)=\outt$.

Consider a sequence of distinct edges $\pi=e_1,\ldots,e_k$ in $G$ such that if we forget about the orientation of edges, then we get a walk $v_1,\ldots,v_{k+1}$ in the underlying undirected graph, where for $i=1,\ldots, k$ edge $e_i$ is an orientation of $v_iv_{i+1}$.
Assume additionally that for every $i=2,\ldots,k$ either both edges $e_{i-1}$ and $e_i$ enter $v_i$ or both leave $v_i$, in other words, the orientation of edges on the walk alternates.
Now, let $I_1, \ldots, I_{k+1}$ be the consecutive interfaces visited by $\pi$, i.e., for every $j=1,\ldots,k+1$ we have that $I_j$ is an interface of $v_j$ and for every $j=1,\ldots,k$, we have $e_j\in I_j\cap I_{j+1}$.
If $|I_1|, |I_k| > 2$ and $|I_j| = 2$, for $j = 2,\ldots, k-1$,
the~sequence~$\pi$ will be called a~\emph{switching walk}.
Similarly, if $|I_j| = 2$ for $j= 1, \ldots, k$, and $v_1=v_{k+1}$, i.e., the walk $v_1,\ldots,v_{k+1}$ is closed, then $\pi$ will be called a~\emph{switching circuit}.
In both cases, {\em length} of~$\pi$ is defined as~$k$.
The sequence $v_1,\ldots,v_{k+1}$ is called the {\em vertex sequence} of $\pi$.
Abusing the notation slightly, we will refer to~$\pi$ as a~set, when it is convenient.
The~motivation for introducing the notions of switching walks and circuits is given by the~following lemma.

\begin{lemma}
\label{lem:switching}
Let $\pi = \{e_1, \ldots, e_k\}$ be a~switching walk or a~switching circuit in a~digraph~$G$.
Let $H \subseteq E(G)$ be a~Hamiltonian cycle in~$G$.
Then, $H \cap \pi = \{e_{2i-1} \mid i = 1, \ldots, \lfloor \frac{k + 1}{2} \rfloor\}$,
or $H \cap \pi = \{e_{2i} \mid i = 1, \ldots, \lfloor \frac{k}{2} \rfloor \}$.
\end{lemma}

\begin{proof}
  Let us assume that $\pi$ is a~switching walk. (For a~switching circuit the~proof is analogous.)
  Consider two consecutive edges $e_i, e_{i+1} \in \pi$.
  By the definition of a~switching walk, there is a~vertex $v$ with an~interface $I$ of size $2$
  such that $I = \{e_i, e_{i+1}\}$.
  Since the~cycle $H$ passes through~$v$, we obtain that $H$ must contain exactly one of the edges $e_i$ and $e_{i+1}$,
  and the lemma easily follows.
\end{proof}

In some cases it is convenient to study switching walks and circuits in the language of an auxiliary bipartite graph.
Let $V^\outt=\{v^\outt \mid v\in V(G)\}$ and $V^\inn=\{v^\inn \mid v\in V(G)\}$.
The {\em interface graph of $G$} is the bipartite graph $I_G$ such that $V(I_G)=V^\outt\cup V^\inn$ and
$E(I_G)=\{u^\outt v^\inn\mid (u,v)\in E(G)\}$.
Clearly, there is a one-to-one correspondence between interfaces in $G$ and vertices of $I_G$, and the degree of a vertex in $I_G$ is the size of the corresponding interface.
Moreover, if $\pi=e_1,\ldots,e_k$ is a switching walk in $G$ with a vertex sequence  $v_1,\ldots,v_{k+1}$ and interface sequence $I_1,\ldots,I_{k+1}$, then $\pi$ corresponds to a simple path $I(\pi)=v_1^{\type(I_1)},\ldots,v_{k+1}^{\type(I_{k+1})}$ in $G$ with endpoints of degree larger than 2, and all inner vertices of degree 2.
Similarly, a switching circuit $\pi$ corresponds to a simple cycle $I(\pi)$ in $I_G$ with all vertices of degree 2 in $I_G$, i.e., $I(\pi)$ forms a connected component in $I_G$.
Observe that both in the case of path and cycle above, the edges $I(\pi)$ are exactly the edges of $I_G$ corresponding to the edges of $\pi$. Using the equivalence described in this paragraph, the following lemma is immediate.

\begin{lemma}
	\label{lem:partition}
	Edges of every digraph can be uniquely partitioned into switching walks and circuits.
	Moreover, the partition can be computed in linear time.
\end{lemma}

\begin{proof}
	Let $G$ be a digraph.
	Recall that by the definition of $I_G$, there is a one-to-one correspondence between edges of $G$ and edges of $I_G$.
	It is clear that edges of $I_G$ can be uniquely partitioned into (1) cycles with all vertices of degree 2 and (2) paths with both endpoints of degree at least 3 and all inner vertices of degree 2.
	The corresponding switching circuits and switching walks form the desired partition of $E(G)$.
	An algorithm which constructs the partition is straightforward.
\end{proof}

\heading{Another view on Corollary~\ref{cor:2-regular}}
The run of the algorithm from Corollary~\ref{cor:2-regular}
can be interpreted using the introduced notions as follows.
We apply Lemma~\ref{lem:partition} to partition~$E(G)$ into switching walks and circuits.
If there is a switching circuit~$\pi$ of length at least~$6$,
we guess the~intersection of~$\pi$ with a~hypothetical Hamiltonian cycle in $G$.
By~Lemma~\ref{lem:switching} there are two possibilities for this~intersection
(in both cases it consists of at least~$3$ edges of $G$).
We consider both cases by marking chosen edges as forced, and recursively calling
on the remaining graph.
If there is no switching circuit of length at least~$6$, it turns out that the remaining instance
can be solved by finding a~minimum spanning tree in an~auxiliary graph (see \cite{Eppstein:cubic} for the details).
Hence, the size of the recursion tree of this algorithm can be bounded by $2^{n/3}$.

\subsection{Branching subroutine}
\label{subsec:avgdeg-branching}

Let us consider a~digraph~$G$.
By $t_i(G)$ we will denote the number of vertices of~$G$ with outdegree equal to~$i$.
Let~$k = n - t_1(G)$ be the number of vertices of $G$ with outdegree at least~$2$,
and let $s_1, \ldots, s_k$ be the sequence of these outdegrees.
Then, let us denote the sum ${\sum_{i=1}^k (s_i - 2)}$ by~$S(G)$.
An~analogous sum for indegrees will be denoted by $S^-(G)$.
Note that if $G$ has no vertex of out- or indegree~$1$, then by the handshaking lemma
$S(G) = S^-(G)$.

\begin{theorem}
\label{thm:branching}
  \probATSP can be solved in time $\O^*(2^{(n - t_1(G))/3\ +\ \beta S(G)})$
  and polynomial space,
  where $\beta = \log_2 3 - 1 < 0.585$.
\end{theorem}

\begin{proof}
{
\newcommand{\rightpar}{)}
\newcommand{\FuncName}[1]{{\sf #1}}
\newcommand{\FBranch}{\FuncName{AtspBranching}}
\newcommand{\FEdge}{\FuncName{AtspForcedEdge}}

\newcommand{\varResult}{{\sf result}}
\newcommand{\varWeight}{{\sf weight}}

\begin{algorithm}[t]
	\caption{$\FBranch(G, \varWeight)$}
	\label{alg:atsp-branching}
	\footnotesize
	\DontPrintSemicolon
	\SetKwProg{Fn}{Function}{:}{}

	\KwIn{$G$ -- a digraph on $n \geq 2$ vertices,\newline
		$\varWeight$ -- a function $E(G) \to \ZZ$}
	\KwOut{the minimum weight of a Hamiltonian cycle in $G$,\newline
		or $\infty$ if there is no such cycle\newline}

	\Fn{$\FEdge(G, \varWeight, e)$}{
		Let $e = (u, v)$\;
		$G_1 \leftarrow G$ with removed edges of the form $(v, u)$, $(u, x)$ and $(x, v)$ for $x \in V(G)$\;
		$G' \leftarrow G_1$ with contracted vertices $u$ and $v$\;
		$\varWeight' \leftarrow$ weights of $E(G')$ inherited from $G$ appropriately\;
		\Return{$\varWeight(e) + \FBranch(G', \varWeight')$}
	}

	\vspace{1em}

	\Fn{$\FBranch(G, \varWeight)$}{
		\If(\hfill$(a\rightpar$){$G$ has exactly two vertices $u$ and $v$}{
			\Return{$\varWeight((u, v)) + \varWeight((v, u))$ if $(u, v), (v, u) \in E(G)$, or $\infty$ otherwise}
		}
		\If(\hfill$(b\rightpar$){there is an empty interface in $G$ i.e. a vertex of out- or indegree $0$}{
			\Return{$\infty$}
		}
		\If(\hfill$(c\rightpar$){there is an interface $I = \{e\}$ of size $1$}{
			\Return{$\FEdge(G, \varWeight, e)$}
		}

		Use Lemma~\ref{lem:partition} to partition $E(G)$ into switching walks and circuits\;

		\If(\hfill$(d\rightpar$){there is a switching walk $\pi$ which begins and ends at the same interface $I$}{
			$G' \leftarrow G$ with removed edges of $I \setminus \pi$\;
			\Return{$\FBranch(G', \varWeight)$}
		}

		\If(\hfill$(e\rightpar$){there is a switching walk $\pi$ of even length}{
			Let $\pi = (e_1, \ldots, e_{2k})$\;
			\Return{$\min (\FEdge(G, \varWeight, e_1), \FEdge(G, \varWeight, e_{2k}))$}
		}
		\If(\hfill$(f\rightpar$){there is no interface of size at least $3$}{
			Apply Corollary~\ref{cor:2-regular} to $G$ and return the weight of the solution, or $\infty$\;
		}
		\Else(\hfill$(g\rightpar$){
			Let $I = \{e_1, \ldots, e_s\}$ be an out-interface of size $s \geq 3$\;
			$\varResult \leftarrow \infty$\;
			\For{$i = 1, \ldots, s$}{
				$\varResult \leftarrow \min(\text{result}, \FEdge(G, \varWeight, e_i))$\;
			}

			\Return{\varResult}
		}
	}
\end{algorithm}

The idea behind this algorithm is to branch on interfaces of size greater than~$2$,
reducing the initial problem to the case of $(2,2)$-graphs, and then to apply Corollary~\ref{cor:2-regular}.
A~detailed description is presented in Pseudocode~\ref{alg:atsp-branching}.
Our algorithm consists of two functions: $\FBranch(G, \varWeight)$ -- the main one,
which solves \probATSP in~$G$,
and an~auxiliary function $\FEdge(G, \varWeight, e)$ that returns the minimum weight of a~Hamiltonian cycle~$H$ in~$G$
such that~$e \in H$ (or $\infty$ if there is no such cycle).
Note that $\FEdge$ modifies the input digraph~$G$, and calls $\FBranch$ on the new digraph~$G'$.
We observe that every Hamiltonian cycle in~$G'$ of weight~$w$
corresponds to a~Hamiltonian cycle in $G$ of weight $w + \varWeight(e)$
and containing edge~$e$, and vice versa.

Given a~digraph~$G$ with a~function $\varWeight : E(G) \to \ZZ$, $\FBranch$
starts by considering a~number of trivial cases $(a) - (c)$,
where either $G$~has only $2$~vertices, or there is a~vertex with out- or indegree at most~$1$.
Next, we apply Lemma~\ref{lem:partition} to decompose $E(G)$ into switching walks and circuits, and we deal with
a~situation when there is a~switching walk~$\pi = (e_1, \ldots, e_{2k})$
of even length in~$G$ (cases $(d) - (e)$ in Pseudocode~\ref{alg:atsp-branching}).
Denote by~$I$, respectively $I'$, the interface which $\pi$~starts, respectively ends, at.
Consider a~Hamiltonian cycle~$H$ in~$G$.
By Lemma~\ref{lem:switching} we obtain that either $e_1 \in H \cap \pi$,
or $e_{2k} \in H \cap \pi$.
We consider the following two cases.
\begin{itemize}
  \item If $I = I'$, then we have $H \cap I \in \{e_1, e_{2k}\}$,
  and thus all edges of $I \setminus \pi$ can be safely removed
  as they cannot be extended to a~Hamiltonian cycle in~$G$.
  This is realized in step~$(d)$ of the pseudocode.
  Note that if a~switching walk~$\pi$ starts and ends at the same interface, then it must be of even length,
  since orientation of edges on $\pi$ alternates.

  \item If $I \neq I'$, we branch by guessing if $e_1 \in H \cap \pi$, or $e_{2k} \in H \cap \pi$
  (step $(e)$ of the pseudocode).
\end{itemize}
If none of the above cases holds, we check whether all interfaces consist of at most~$2$ edges (cases $(f) - (g)$ in Pseudocode~\ref{alg:atsp-branching}).
If so, then $G$ is a~$(2,2)$-graph, and we can solve \probATSP for $G$ by applying Corollary~\ref{cor:2-regular}.
If not, we choose an~out-interface~$I$ of size at least $3$, and we branch on it, by guessing which
of the edges of~$I$ to pick as a~part of a~Hamiltonian cycle.
Note that since $G$ has no interface of size~$1$, then it has an~interface of size at least~$3$
if and only if it has an~out-interface of size at least~$3$.

\heading{Time complexity analysis}
We begin with providing a~few simple facts concerning the properties of our algorithm.

\begin{claim}
\label{claim:alg-potential}
  During execution of algorithm $\FBranch$, the value of $S(G)$ cannot increase.
\end{claim}

{
\renewcommand\qedsymbol{$\lrcorner$}
\begin{proof}
  Clearly, removing an edge cannot increase the value of $S(G)$.
  Moreover, whenever we contract an~edge $(u, v)$ (call the resulting vertex $x$),
  we remove edges of the form $(v, u)$, $(u, \underline{\hspace{0.5em}})$, $(\underline{\hspace{0.5em}}, v)$.
  Hence, the out-interface, respectively the in-interface, of $x$ is a~subset
  of the out-interface of $v$, respectively the in-interface of $u$,
  and the other interfaces remain unchanged.
\end{proof}
}

\begin{claim}
\label{claim:alg-simple}
  During execution of algorithm $\FBranch$, graph~$G$ is simple, i.e. does not contain
  two edges of the same head and tail.
\end{claim}

{
\renewcommand\qedsymbol{$\lrcorner$}
\begin{proof}
  Without loss of generality, we may assume that the input graph is simple,
  for otherwise we just discard the lighter edge.
  Moreover, before contracting an~edge $(u, v)$ (call the resulting vertex $x$),
  we remove edges of the form $(v, u)$, $(u, \underline{\hspace{0.5em}})$, $(\underline{\hspace{0.5em}}, v)$.
  Hence, after contraction there is no loop $(x, x)$,
  every edge outgoing from $x$ corresponds to an~edge outgoing from~$v$,
  and every edge incoming to $x$ corresponds to an~edge incoming to~$u$.
\end{proof}
}

\begin{claim}
\label{claim:alg-switching}
  Let $\pi = (e_1, \ldots, e_k)$ be a~switching walk in~$G$.
  Assume that during the run of our algorithm we decided to take an~edge~$e_1$ by calling $\FEdge(G, \varWeight, e_1)$.
  Then, by exhaustively applying rule $(c)$ of $\FBranch$ to the resulting digraph,
  we will remove from~$G$ all edges of the form $e_{2i}$, and contract all edges of the form $e_{2i+1}$.
  An~analogous statement can be made if we start with discarding edge~$e_1$ instead of contracting it.
  \hfill$\lrcorner$
\end{claim}

Denote $f(n, S) = 2^{n/3 + \beta S}$, where $\beta$ is the constant from Theorem~\ref{thm:branching}.
We need to prove that the running time of our algorithm is bounded by $f(n - t_1(G), S(G)) n^{\Oh(1)}$.
We proceed by induction on $t_1(G) + S(G)$.

If $t_1(G) > 0$, then our algorithm
starts by choosing edges which form interfaces of size~$1$,
what leads to a~digraph with at most $\max(2, n - t_1(G))$ vertices.
Hence, by the induction hypothesis the running time is bounded by $f(n - t_1(G), S(G)) n^{\Oh(1)}$.

In what follows we assume $t_1(G) = 0$.
If $G$~satisfies condition $(a)$ or $(b)$, then our algorithm runs in polynomial time.
Similarly, we can assume that $G$ does not satisfy conditions $(c)$ and~$(d)$,
as applying the corresponding reductions exhaustively takes only
polynomial time and does not increase the value of $S(G)$,
according to Claim~\ref{claim:alg-potential}.

From now on, we assume that conditions $(a) - (d)$ do not hold for $G$.
If $S(G) = 0$, then our algorithm
executes the algorithm from Corollary~\ref{cor:2-regular} and therefore
its running time is bounded by $\Ohstar(2^{n/3})$, as desired.
Now, assume $S(G) > 0$.
It remains to analyse cases $(e)$ and $(g)$ of $\FBranch$.

\heading{Case $(e)$} Let us assume that there is a~switching walk $\pi = (e_1, \ldots, e_{2k})$
of even length in~$G$ which starts at interface~$I$ of size $s \geq 3$,
and ends at interface~$I' \neq I$ of size $s' \geq 3$.
Let $G'$ be a~digraph obtained from~$G$ by running $\FEdge(G, \varWeight, e_1)$
and exhaustively applying rules $(a) - (d)$ to the resulting digraph.

Since edge $e_1$ is contracted in $\FEdge$, we have $|V(G')| \leq |V(G)| - 1$.
We claim that $S(G') \leq S(G) - 2$.
Assume $\type(I) = \type(I') = \outt$.
By Claim~\ref{claim:alg-switching}, for all $i = 1,\ldots,k$, edge $e_{2i-1}$ was contracted,
and edge $e_{2i}$ was removed.
We observe that contracting edge $e_1$ results in removing interface~$I$
from the graph, and discarding edge $e_{2k}$ decreases the size of~$I'$ by~$1$.
By Claim~\ref{claim:alg-potential} operations performed on edges $e_2, \ldots, e_{2k-1}$ do not increase the value of~$S(G)$.
Hence, $S(G) - S(G') \geq (s - 2) + 1 \geq 2$,
as desired.
If $\type(I) = \type(I') = \inn$, then by the same reasoning, we obtain $S^-(G') \leq S(G) - 2$
but since there are no interfaces of size~$1$ in~$G'$, we have $S(G') = S^-(G')$,
and the claim follows.

Hence, by the induction hypothesis, the running time of our algorithm applied to $G'$
is bounded by $f(n - 1, S(G) - 2)$.
To obtain the desired bound for digraph~$G$
we need to show that $2f(n-1, S(G)-2) \leq f(n, S(G))$, or, equivalently $\log_2(2f(n-1, S(G)-2)) \leq \log_2 f(n, S(G))$. We obtain
\begin{align*}
  \log_2(2f(n-1, S(G)-2)) & = 1 + \tfrac{n-1}{3} + \beta (S(G)-2)
  = \tfrac{n}{3} + \beta S(G) + \tfrac{2}{3} - 2\beta \\
  & \leq \tfrac{n}{3} + \beta S(G)
  = \log_2 f(n, S(G)).
\end{align*}

\heading{Case $(g)$}
Now, we assume that $G$ does not satisfy conditions $(a)-(f)$.
Let $I$ be an~out-interface of size $s \geq 3$, and consider an~edge~$e \in I$.
Let $G'$ be a~digraph obtained from~$G$ after choosing edge~$e$
by running $\FEdge(G, \varWeight, e)$,
and let $G''$ be a~digraph obtained from~$G'$
by the subsequent exhaustive application of rules ${(a)-(d)}$ by $\FBranch$.
Define $\Delta n = |V(G)| - |V(G'')|$, and ${\Delta S = S(G) - S(G'')}$.

\begin{claim}
\label{lem:potential}
  It holds that $\Delta n \geq 1$, $\Delta S \geq s - 2 \geq 1$, and $\Delta n + \Delta S \geq s + 1$.
\end{claim}

{
\renewcommand\qedsymbol{$\lrcorner$}
\begin{proof}
  For a~digraph $G$ we denote $n(G) = |V(G)|$.
  First, we analyse a~direct impact of calling $\FEdge(G, \varWeight, e)$.
  All edges of~$I$ are removed from $G$, hence by Claim~\ref{claim:alg-potential}
  we have $\Delta S \geq S(G) - S(G') \geq s - 2 \geq 1$.
  Moreover, edge~$e$ gets contracted, and thus $\Delta n \geq n(G) - n(G') = 1$.
  We are left with proving that $(n(G') - n(G'')) + (S(G') - S(G'')) \geq 2$,
  since then we will have
  \begin{align*}
    \Delta n + \Delta S & = (n(G) - n(G') + n(G') - n(G'')) + (S(G) - S(G') + S(G') - S(G'')) \\
    & = (n(G) - n(G')) + (S(G) - S(G')) + (n(G') - n(G'') + S(G') - S(G'')) \\
    & \geq 1 + (s - 2) + 2 = s + 1.
  \end{align*}

  Let $\pi$ be the~switching walk which starts with edge~$e$.
  Let $e'$~be the last edge of~$\pi$ (it is possible that $\pi$ has length~$1$ and $e' = e$).
  We recall that at step~$(g)$ every switching walk in $G$ is of odd length.
  Take an~in-interface~$I'$ such that $e' \in I'$.
  By the definition of switching walk, $|I'| \geq 3$, so let $e', e_1', e_2'$
  be three different edges of~$I'$.
  For $j = 1,2$ denote by $\pi_j$ the~switching walk which ends with edge~$e_j'$.
  Let $e_j$ be the first edge of~$\pi_j$, and let~$I_j$ be an~out-interface such that~$e_j \in I_j$.

  Let $F, R \subseteq E(G)$ be edges of~$G$ which correspond to the edges that were taken (and hence, contracted)
  and removed, respectively, during the run of our algorithm which leads from digraph~$G$
  to digraph~$G''$.
  We have $e \in F$. By Claim~\ref{claim:alg-switching} applied to~$\pi$, we obtain $e' \in F$. Therefore, $e_1', e_2' \in R$,
  and again by Claim~\ref{claim:alg-switching} applied to~$\pi_1$ and $\pi_2$, we obtain $e_1, e_2 \in R$.
  Now, we consider a~few cases.
  \begin{itemize}
    \item If $I, I_1, I_2$ are pairwise different out-interfaces, then during processing of digraph~$G'$
    we removed edges $e_1, e_2$ from different out-interfaces of size at least~$3$.
    Therefore, $S(G') - S(G'') \geq 2$.

    \item If $I = I_1 = I_2$, then among switching walks $\pi, \pi_1, \pi_2$ there are least two
    of length greater than~$1$ (hence, of length at least~$3$),
    because otherwise the graph is not simple, contradicting Claim~\ref{claim:alg-simple}.
    Let us assume that these are
    walks $\pi$ and $\pi_1$ (the other cases are analogous).
    Then, by Claim~\ref{claim:alg-switching}, during processing of digraph~$G'$ we contracted edge $e'$
    and the second edge of~$\pi_1$.
    Therefore, $n(G') - n(G'') \geq 2$.

    \item If $I_1 = I_2 \neq I$, or $I = I_1 \neq I_2$, or $I = I_2 \neq I_1$,
    then at least one switching walk among~$\pi, \pi_1, \pi_2$ is of length at least~$3$,
    and there is another interface apart from~$I$ that gets smaller.
    Hence, we obtain in a similar way as before that
    $n(G'') - n(G') \geq 1$, and $S(G'') - S(G') \geq 1$.
  \end{itemize}
\end{proof}
}

Since $\Delta S \geq 1$, we have $S(G'') < S(G)$, and thus by the induction hypothesis
the running time of our algorithm applied to $G''$ is bounded by~$f(n(G''), S(G'')) = f(n - \Delta n, S(G) - \Delta S)$.
In step~$(g)$ of~$\FBranch$ we branch into $s$~such subcases, hence
we need to prove that $s\cdot f(n - \Delta n, S(G) - \Delta S)) \leq f(n, S(G))$.
We will show the equivalent $\log_2(s \cdot f(n - \Delta n, S(G) - \Delta S)) \leq \log_2 f(n, S(G))$.
Indeed,
\begin{align*}
  \log_2(s \cdot f(n - \Delta n, S(G) - \Delta S))
  & = \log_2 s + \tfrac{n - \Delta n}{3} + \beta (S(G) - \Delta S) \\
  & = \tfrac{n}{3} + \beta S(G) + \log_2 s - \tfrac{\Delta n}{3} - \beta \Delta S \\
  & \leq \tfrac{n}{3} + \beta S(G) + \log_2 s - \tfrac{s + 1 - \Delta S}{3} - \beta \Delta S && \text{(Claim \ref{lem:potential})}\\
  & = \tfrac{n}{3} + \beta S(G) + \log_2 s - \tfrac{s + 1}{3} - (\beta - \tfrac{1}{3}) \Delta S \\
  & \leq \tfrac{n}{3} + \beta S(G) + \log_2 s - \tfrac{s + 1}{3} - (\beta - \tfrac{1}{3})(s - 2) && \text{(Claim \ref{lem:potential})}\\
  & = \tfrac{n}{3} + \beta S(G) + \log_2 s - 1 - \beta (s - 2) \\
  & \leq \tfrac{n}{3} + \beta S(G) && (\bigtriangleup) \\
  & = \log_2 f(n, S(G)).
\end{align*}
where inequality~$(\bigtriangleup)$ follows from the fact that
the function $x \mapsto \frac{\log_2 x - 1}{x-2}$
is decreasing on~$[3, \infty)$, and thus it can be bounded by the value at $x = 3$ which is equal to $\beta$.
Consequently, the inequality $\log_2 s \leq 1 + \beta (s - 2)$ holds for $s \geq 3$.

}

\end{proof}

\subsection{General algorithm}

The idea behind our general algorithm is to run in parallel two algorithms:
our branching algorithm from Theorem~\ref{thm:branching} (which we will refer to as Algorithm~\encircle{a}),
and enumerating cycle covers from Subsection~\ref{subsec:reductions-bregman} (Algorithm~\encircle{b}) .
We finish when one of these algorithms terminates.
Our goal is to prove that the time complexity of such an~approach is bounded by
$\Ohstar(2^{\alpha(d-1)n})$ if we apply it to digraphs of average outdegree at most~$d$,
where $\alpha$ is the constant from Theorem~\ref{thm:avgdeg}.

Note that when implementing this algorithm, one may also compare the values of ${\frac{n}{3}+\beta S(G)}$
and $\alpha(d-1)n$, and, depending on the result, run either \algoref{a}, or \algoref{b}.

Let $G$ be a~digraph on $n$~vertices, of average outdegree~$d$.
We may assume that $d > 1$, for otherwise \probATSP in~$G$
can be solved in polynomial time.
Let $t_i := t_i(G)$ for $i = 1, \ldots, n-1$,
and denote $\vect{t} = (t_1, \ldots, t_{n-1})$.
The numbers $t_1, \ldots, t_{n-1}$ satisfy the conditions
\begin{equation}
\label{equ:ti-conditions}
\begin{cases}
  \sum_{i=1}^{n-1} t_i & = n \\
  \sum_{i=1}^{n-1} t_i\cdot i & = d n
\end{cases}
\tag{$\lozenge$}
\end{equation}

According to Theorem~\ref{thm:branching} and Corollary~\ref{cor:bregman-matchings},
the running time of \algoref{a} and \algoref{b} for~$G$,
up to a polynomial factor, can be bounded by functions
\[
  f(\vect{t}) := 2^{\frac{n-t_1}{3} + \beta (\sum_{i=2}^{n-1} t_i (i - 2))}
\]
and
\[
  g(\vect{t}) := \prod_{i=2}^{n-1} (i!)^{t_i / i},
\]
respectively.
Define $h(\vect{t}) = \min(f(\vect{t}), g(\vect{t}))$.
Our task can be rephrased in the following way: we want to find
the maximum value of $h(\vect{t})$
over vectors $\vect{t}$ which satisfy conditions~(\ref{equ:ti-conditions}).
From now on, we assume that the numbers $t_1, \ldots, t_{n-1}$ satisfy these conditions
and maximize the value of $h(\vect{t})$.
If there are many such valuations, then we choose the one which is lexicographically minimal.

\begin{claim}
\label{claim:ti}
  If $t_i, t_j \neq 0$, for $i, j \geq 2$, then $|i - j| \leq 1$.
  In particular, among numbers $t_2, \ldots, t_n$ at most two are nonzero.
\end{claim}

{
\renewcommand\qedsymbol{$\lrcorner$}
\begin{proof}
  Suppose that $t_i, t_j > 0$ where $i < i + 2 \leq j$.
  We define a~vector $\vect{t'}$ such that $t_k' = t_k$ for all $k$
  except for $k = i, i+1, j-1, j$ where $t_k'$ is equal to $t_i-1, t_{i+1}+1, t_{j-1}+1, t_j-1$,
  respectively.
  (However, we set $t_{i+1}' = t_{i+1} + 2$ if $i + 1 = j - 1$.)
  We observe that conditions~(\ref{equ:ti-conditions}) still hold for $\vect{t'}$.
  Moreover, we have $f(\vect{t'}) = f(\vect{t})$, and by Lemma~\ref{lem:bregman-inequality}
  $g(\vect{t'}) \geq g(\vect{t})$.
  Hence, $h(\vect{t'}) \geq h(\vect{t})$, and $\vect{t'} <_{\text{lex}} \vect{t}$
  which contradicts the choice of vector~$\vect{t}$.
\end{proof}
}

In further analysis we denote $r_i := t_i / ((d-1)n)$ for $i=1, \ldots, n-1$.
By Claim~\ref{claim:ti}
it is enough to consider the following two cases.

\heading{Case 1} $t_i = 0$ for $i \not\in \{1, 2, 3\}$.

Then, conditions~(\ref{equ:ti-conditions}) take the form of $t_1 + t_2 + t_3 = n$
and $t_1 + 2t_2 + 3t_3 = dn$.
Hence, $t_2 + 2t_3 = (d-1)n$, or, equivalently,
$r_2 + 2r_3 = 1$.  Thus, we have
\begin{align*}
  \tfrac{1}{(d-1)n}\log_2 f(\vect{t}) = \tfrac{1}{3}(r_2 + r_3) + \beta r_3 = \tfrac{1}{3}(1 - 2r_3 + r_3) + \beta r_3
  = \tfrac{1}{3} + \left(\beta - \tfrac{1}{3}\right)r_3 \\
  \tfrac{1}{(d-1)n}\log_2 g(\vect{t}) = \tfrac{1}{2}r_2 + \tfrac{\log_2 6}{3} r_3
  = \tfrac{1}{2}(1 - 2r_3) + \tfrac{\log_2 6}{3} r_3 = \tfrac{1}{2} - \left(1 - \tfrac{\log_2 6}{3}\right) r_3.
\end{align*}

Hence, we obtain that $f(\vect{t}) = 2^{\frac{1}{3}(d-1)n}$ if $r_3 = 0$, and $f(\vect{t})$ is increasing
as a~function of~$r_3$. Similarly, $g(\vect{t})$ is a~decreasing function of~$r_3$ with $g(\vect{t}) = 2^{\frac{1}{2}(d-1)n}$
for $r_3 = 0$.
Therefore, the minimum of $f(\vect{t})$ and $g(\vect{t})$ can be upper bounded by the value of $f$ at point $\vect{t_0}$
such that $f(\vect{t_0}) = g(\vect{t_0})$. In our case this equality holds if
\[
  \tfrac{1}{3} + \left(\beta - \tfrac{1}{3}\right)r_3 = \tfrac{1}{2} - \left(1 - \tfrac{\log_2 6}{3}\right) r_3
\]
from which we obtain
\[
  r_3 = \frac{\frac{1}{2} - \frac{1}{3}}{\beta - \frac{1}{3} + 1 - \frac{\log_2 6}{3}}
  = \frac{\frac{1}{6}}{\beta - \frac{1}{3} + 1 - \frac{\log_2 3 + 1}{3}} = \frac{\frac{1}{6}}{\beta - \frac{\log_2 3 - 1}{3}}
  = \frac{\frac{1}{6}}{\beta - \frac{\beta}{3}} = \frac{1}{4\beta}.
\]
To obtain a~bound on $h(\vect{t})$ it remains to plug the above value into the formula for $f(\vect{t})$:
\[
  \tfrac{1}{(d-1)n}\log_2 h(\vect{t}) \leq \tfrac{1}{3} + \left(\beta - \tfrac{1}{3}\right)\tfrac{1}{4\beta}
  = \tfrac{7}{12} - \tfrac{1}{12\beta} = \alpha
\]
Hence, $h(\vect{t}) \leq 2^{\alpha(d-1)n}$, as desired.

\heading{Case 2} $t_i = 0$ for $i \not\in \{1, s, s+1\}$, where $s \geq 3$.

We will show that in this case $g(\vect{t}) = o(2^{0.431(d-1)n}) = o(2^{\alpha(d-1)n})$.
Conditions~(\ref{equ:ti-conditions}) take now the form of $t_1 + t_s + t_{s+1} = n$ and $t_1 + st_s + (s+1)t_{s+1} = dn$.
Hence, we obtain $(s-1)t_s + st_{s+1} = (d-1)n$, or, equivalently, $(s-1)r_s + sr_{s+1} = 1$. Therefore
\begin{align*}
  \tfrac{1}{(d-1)n}\log_2 g(\vect{t}) & = \tfrac{\log_2 s!}{s} r_s + \tfrac{\log_2 (s+1)!}{s+1} r_{s+1}
  = \tfrac{\log_2 s!}{s(s-1)} (1 - sr_{s+1}) + \tfrac{\log_2 (s+1)!}{s+1} r_{s+1} \\
  & = \tfrac{\log_2 s!}{s(s-1)} - \left(\tfrac{\log_2 s!}{s - 1} - \tfrac{\log_2 (s+1)!}{s + 1}\right)r_{s+1}.
\end{align*}

One can verify that the inequality $\frac{\log_2 s!}{s - 1} > \frac{\log_2 (s+1)!}{s + 1}$ holds for $s \geq 2$,
and that the sequence $\left(\frac{\log_2 s!}{s(s-1)}\right)_{s \geq 3}$ is decreasing. Therefore
\[
  \tfrac{1}{(d-1)n}\log_2 g(\vect{t}) \leq \tfrac{\log_2 6}{6} < 0.431.
\]
Hence, $g(\vect{t}) < 2^{0.431(d-1)n}$, which ends the proof.

\section{Exponential space algorithm}
\label{sec:gebauer}
In this section we establish Theorem \ref{thm:gebauer_sparse}.

Let $G$ be a~digraph with $n$~vertices and $m = dn$~edges.
For simplicity, we assume in this section that $n$ is even,
for otherwise we can pick an~arbitrary vertex~$v$, and split
it into two vertices $v^\inn$ and $v^\outt$
with edges inherited from $v$ appropriately
and with one additional edge $(v^\inn, v^\outt)$
-- this operation
adds one vertex to the graph but does not increase the average
outdegree.
We will say that a~simple path~$P$ in~$G$ is \mbox{\emph{$(l,D)$-light}}
if~the~length of~$P$ is~$l$,
and the sum of outdegrees of inner vertices of~$P$ is bounded by~$D$.
For a~vertex $v \in V(G)$, and positive integers $l$,~$D$,
by~$\mathcal{P}_{v, l, D}$ we will denote
the family of all $(l,D)$-light paths in~$G$ which start at vertex~$v$.

Our algorithm relies on the following two lemmas.

\begin{lemma}\label{lem:circle}
  Let $H$ be a~Hamiltonian cycle in $G$. Then, the edges of $H$ can be partitioned
  into two $(n/2,\ m/2)$-light paths.
\end{lemma}

\begin{lemma}\label{lem:generate_paths}
  For a~digraph $G$, a vertex $v$, and integers $l$, $D$,
  the family $\mathcal{P}_{v, l, D}$ can be computed in time $\tau(D/(l-1))^{l-1} n^{\Oh(1)}$
  where the~function $\tau$ is defined as in Theorem~\ref{thm:gebauer_sparse}.
\end{lemma}

Before we proceed to the proofs of above lemmas,
let us see how to derive Theorem~\ref{thm:gebauer_sparse} from them.
Given a~digraph~$G$, the algorithm starts by iterating over all pairs of distinct vertices
$u_1$ and $u_2$.
For each such a~pair we use Lemma~\ref{lem:generate_paths} to obtain
the families $\mathcal{P}_1 = \mathcal{P}_{u_1, n/2, m/2}$
and~$\mathcal{P}_2 = \mathcal{P}_{u_2, n/2, m/2}$.
By filtering them, we may assume that all paths from $\mathcal{P}_1$ end at~$u_2$,
and all paths from $\mathcal{P}_2$ end at~$u_1$.
Next, we create a~dictionary $\mathcal{D}$ with an~entry
$\{\text{key}: V(P_1), \text{value}: \text{weight}(P_1)\}$ for every path $P_1 \in \mathcal{P}_1$.
(In case there is more than one path on the same set of vertices we keep only one~entry
with the minimum weight.)
Then, we iterate over all paths $P_2 \in \mathcal{P}_2$,
and we look up in $\mathcal{D}$ a~subset ${V'(P_2) := (V(G) \setminus V(P_2)) \cup \{u_1, u_2\}}$.
For every hit we calculate the sum: $\text{weight}(P_2)$ + $\mathcal{D}[V'(P_2)]$,
and we return the minimum of these values.

The correctness of this procedure is a~direct corollary from Lemma~\ref{lem:circle}.
Moreover, the running time of the algorithm is dominated, up to a~polynomial factor,
by the running time of the~algorithm from Lemma~\ref{lem:generate_paths}, which
in our case is bounded by
\[
  \tau \left(\frac{\tfrac{m}{2}}{\tfrac{n}{2}-1} \right)^{n/2 - 1} n^{\Oh(1)}
  = \tau \left(\frac{d}{1-\tfrac{2}{n}} \right)^{n/2 - 1} n^{\Oh(1)}
  = \tau(d)^{n/2} n^{\Oh(1)}
\]
where the last equality follows from the fact that when $d$ is fixed,
then for sufficiently large~$n$
we have $\lfloor d/(1-\frac{2}{n}) \rfloor = \lfloor d \rfloor$.
Note that we implement the dictionary~$\mathcal{D}$ as a~balanced tree, so each lookup
takes time $\Oh(\log |\mathcal{D}|) = \Oh(n)$.

\begin{proof}[Proof of Lemma \ref{lem:circle}]
  Let $k = n/2$,
  and let $d_0, d_1, \ldots, d_{2k-1}$ be the outdegrees of consecutive vertices on~$H$.
  Denote $S_i = d_i + d_{i+1} + \ldots + d_{i+k-1}$.
  (In this proof indices are understood modulo $2k$.)
  We need to prove that for some index~$j$ both expressions
  $S_j - d_j$ and $S_{j+k} - d_{j+k}$ do not exceed $m/2$.

  Let $R_i := S_i - S_{i+k}$.
  We observe that $R_k = S_k - S_0 = -R_0$.
  Hence, there exists an~index~$j \in \{0, \ldots, k-1\}$ such that
  $R_j \cdot R_{j+1} \leq 0$.
  Without loss of generality, we may assume that $R_j \leq 0$ (equivalently, $S_j \leq S_{j+k}$),
  for otherwise we can just shift all indices by~$k$.
  Then, $R_{j+1} \geq 0$ (equivalently, $S_{j+1+k} \leq S_{j+1}$).
  Thus we obtain
  \begin{align*}
    S_j - d_j & \leq S_j \leq \tfrac{1}{2}(S_j + S_{j+k}) = \tfrac{m}{2} \\
    S_{j+k} - d_{j+k} & \leq S_{j+k+1} \leq \tfrac{1}{2}(S_{j+k+1} + S_{j+1}) = \tfrac{m}{2}.
  \end{align*}
  This ends the proof.
\end{proof}

Before we proceed to the proof of Lemma \ref{lem:generate_paths}, we state a technical lemma.
  \begin{lemma}
  	\label{lem:tau}
	Let $a_1, \ldots, a_k$ be integers with an average bounded by $\bar{a}$.
	Then, $a_1\cdot\ldots \cdot a_k \leq \tau(\bar{a})^k$.
\end{lemma}

\begin{proof}
	First, we see that if the product of numbers $a_i$ is maximum, then for all $i$ we have
	$a_i = \lfloor \bar{a} \rfloor$ or $a_i = \lceil \bar{a} \rceil$,
	for otherwise we could either increase some number by~$1$, or replace some two numbers $a_i$, $a_j$ ($a_i < a_j$)
	with numbers $a_i + 1$, $a_j - 1$.
	Assume that among numbers $a_1, \ldots, a_k$ there are $k_1$ numbers equal to~$\lfloor \bar{a} \rfloor$
	and $k - k_1$ numbers equal to~$\lceil \bar{a} \rceil$. Moreover, assume that $\bar{a} \not\in \Z$,
	and consequently $\lceil \bar{a} \rceil = \lfloor \bar{a} \rfloor + 1$. Then
	\[
	k_1 \lfloor \bar{a} \rfloor + (k - k_1)(\lfloor \bar{a} \rfloor + 1) = a_1 + \ldots + a_k \leq k \bar{a}
	\]
	and thus $k_1 \geq k(\lfloor \bar{a} \rfloor + 1 - \bar{a})$. Hence
	\[
	a_1\cdot \ldots \cdot a_k
	= {\lfloor \bar{a} \rfloor}^{k_1} {\lceil \bar{a} \rceil}^{k-k_1}
	\leq {\lfloor \bar{a} \rfloor}^{k(\lfloor \bar{a} \rfloor + 1 - \bar{a})}
	(\lfloor \bar{a} \rfloor + 1)^{k(\bar{a} - \lfloor \bar{a} \rfloor)} = \tau(\bar{a})^k.
	\]
	To finish the proof, we observe that if $\bar{a} \in \Z$, then $\tau(\bar{a}) = \bar{a}$,
	so we obtain a~bound $\bar{a}^k$, as desired.
\end{proof}

\begin{proof}[Proof of Lemma \ref{lem:generate_paths}]

  {
  \newcommand{\rightpar}{)}
  \newcommand{\FuncName}[1]{{\sf #1}}
  \newcommand{\FGen}{\FuncName{GeneratePaths}}

  \newcommand{\varPath}{{\sf path}}

  \begin{algorithm}[t]
  \caption{$\FGen(G,\ \varPath,\ l,\ D)$}
  \label{alg:generate_paths}
  \footnotesize

    \KwIn{$G$ -- a digraph,\newline
    $\varPath$ -- a sequence of vertices forming a path in $G$,\newline
    $l$, $D$ -- positive integers}

    \KwOut{A collection of all simple paths of the form: $\varPath \#(v_1, \ldots, v_l)$ such that
    $\sum_{i=1}^{l-1} \outdeg(v_i) \leq D$\newline}

    $u \leftarrow $ the last vertex on $\varPath$\;
    \For{vertex $v_1$ such that $(u, v_1) \in E(G)$ $\mathbf{and}$ $v_1 \not\in \varPath$}{
      \If{$l = 1$}{
        \MyPrint{$\varPath \# v_1$}
      }
      \ElseIf(\hfill $(\checkmark\rightpar$){$\outdeg(v_1) + (l - 2) \leq D$}{
        $\FGen(G,\ \varPath \# v_1,\ l - 1,\ D - \outdeg(v_1))$\;
      }
    }
  \end{algorithm}

  We apply a simple branching procedure which starts at vertex $v$,
  and at each step guesses the next vertex on a~path by considering
  all \emph{reasonable} possibilities.
  A detailed description of the algorithm can be found in Pseudocode \ref{alg:generate_paths}.
  (To compute the family $\mathcal{P}_{v,l,D}$
  we call the function $\FGen$ with the arguments $(G, \{v\}, l, D)$.)
  Note that before appending a~vertex to the current path we check whether
  the sum of outdegrees on the new path is not too large (line marked with $(\checkmark)$ in the Pseudocode).
  More precisely, we check whether appending a~sequence of vertices of outdegree~$1$
  to the new path would give us a~correct $(l,D)$-path.

  The correctness of such a procedure is straightforward.
  It remains to estimate its time complexity.
  Let $\Tree(G)$ be a~search tree representing execution of this algorithm.
  We claim that $\Tree(G)$ contains at most $\tau(D / (l-1))^{l-1}n$ leaves,
  where $\tau(d) = \lfloor d \rfloor ^ {\lfloor d \rfloor + 1 - d} {(\lfloor d \rfloor + 1) ^ {d - \lfloor d \rfloor}}
  \leq d$

  We will say that a~directed rooted tree $\Tree$ has the~property $(\star)$ if for any path
  $u_0, \ldots, u_h$ from the root of $\Tree$ to some leaf $u_h$ we have $h \leq l$, and the sum
  $\sum_{i=1}^{h-1} \outdeg(u_i)$ is bounded by $D - (l - h)$.
  From the description of the algorithm we see that $\Tree(G)$ has the~property~$(\star)$.
  Indeed, let $u_0, \ldots, u_h$ be a~path from the~root to a~leaf in~$\Tree(G)$. Before the algorithm
  entered the~vertex $u_{h-1}$ the following condition was checked:
  \[
    \outdeg(u_{h-1}) + (l - (h-2) - 2) \leq D - \sum_{i=1}^{h-2} \outdeg(u_i)
  \]
  which is equivalent to $\sum_{i=1}^{h-1} \outdeg(u_i) \leq D - (l - h)$.
  Hence it is enough to prove the following claim. (A similar statement appears in the work of Gebauer~\cite{Gebauer:thesis}.)

  \begin{claim}
  \label{claim:leaves}
    Any tree $\Tree$ with the property $(\star)$ has at most $\tau(D / (l-1))^{l-1}n$ leaves.
  \end{claim}

  Given a~tree $\Tree$ with the property $(\star)$ we modify it so that the property $(\star)$ is preserved
  and the number of leaves in it does not increase.
  First, we may assume that all leaves in $\Tree$ are at depth exactly $l$.
  Indeed, let $u_0, \ldots, u_h$ be a~path from the root of $\Tree$ to some leaf $u_h$
  at depth $h < l$. Then, we may append to it a~path $u_h, u_{h+1}, \ldots, u_l$
  -- this operation does not change the number of leaves, and the property $(\star)$
  is preserved because
  \[
  \sum_{i=1}^{l-1} \outdeg(u_i) = \sum_{i=1}^{h-1} \outdeg(u_i) + (l - h) \leq D
  \]

  Next, we modify $\Tree$ iteratively. Let $\Tree_1 := \Tree$. At $i$-th step, for $i = 1, \ldots, l-1$,
  we consider the family $\mathcal{S}_i$ of all subtrees in $\Tree_i$ with a~root at depth $i$.
  Let $S_i \in \mathcal{S}_i$ be a~subtree with the maximum number of leaves.
  We create a~tree~$\Tree_{i+1}$ by substituting in~$\Tree_i$
  all subtrees from $\mathcal{S}_i$ with $S_i$.
  We observe that for every $i=1,\ldots, l-1$ tree~$\Tree_i$ has depth~$l$, the number of leaves in~$\Tree_i$
  is bounded by the number of leaves in $\Tree_{i+1}$,
  and all vertices in~$\Tree_i$ at the same depth $j \leq i-1$
  have the same outdegree.
  Combining the latter property with the~fact that the condition~$(\star)$ holds for leaves in the~subtree~$S_i$,
  we obtain inductively that every tree~$\Tree_i$ still has the property~$(\star)$.

  Now, we consider the~tree $\Tree_l$.
  For $i = 0, \ldots, l-1$ let $d_i$ be the outdegree of any vertex at depth~$i$ in~$\Tree_l$.
  Then we may bound the number of leaves in $\Tree_l$ by
  \[
    d_0 \cdot \prod_{i=1}^{l-1} d_i \leq n \left( \frac{\sum_{i=1}^{l-1} d_i}{l - 1} \right)^{l-1}
    \leq n \left(\frac{D}{l - 1} \right)^{l-1}
  \]

  To obtain a~tighter bound on the size of $\Tree_l$ we observe
  that in the above estimation we obtain an~equality only if $d_i = D / (l-1)$
  for $i = 1, \ldots, l-1$. However, this is impossible unless expression $D / (l-1)$ is integral. After applying Lemma~\ref{lem:tau} we get the tighter bound which proves the claim of Lemma~\ref{lem:generate_paths}.
  }
  \end{proof}

\bibliographystyle{plain}

\end{document}